\documentclass[
  submission,
  copyright,
]{eptcs}

\providecommand{\event}{GandALF 2026}

\title{Antichains for Concurrent Parameterized Games\thanks{This work
is partially supported by ANR BisoUS (ANR-22-CE48-0012).}}
\def\titlerunning{Antichains for Concurrent Parameterized Games}

\def\orcidID#1{\smash{\href{http://orcid.org/#1}{\protect\raisebox{-1.25pt}{\protect\includegraphics{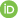}}}}}
\author{
  Nathalie Bertrand\orcidID{0000-0002-9957-5394}
  \institute{Université de Rennes,
  Inria,\\ CNRS, IRISA, Rennes, France}
  \email{nathalie.bertrand@inria.fr}
  \and
  Patricia Bouyer\orcidID{0000-0002-2823-0911}
  \institute{LMF, CNRS \& ENS Paris-Saclay,\\Université Paris-Saclay,\\
  Gif-sur-Yvette, France}
  \email{bouyer@lmf.cnrs.fr}
  \and
  Gaëtan Staquet\orcidID{0000-0001-5795-3265}
  \institute{Nantes Université,\\École Centrale Nantes,\\CNRS, LS2N,
  Nantes, France}
  \email{gaetan.staquet@ec-nantes.fr}
}
\def\authorrunning{N. Bertrand, P. Bouyer, and G. Staquet}

\usepackage{amsthm}
\usepackage{mathtools}
\usepackage{thmtools}

\theoremstyle{plain}
\newtheorem{theorem}{Theorem}[section]
\newtheorem{lemma}[theorem]{Lemma}
\newtheorem{proposition}[theorem]{Proposition}
\newtheorem{corollary}[theorem]{Corollary}

\theoremstyle{definition}
\newtheorem{definition}[theorem]{Definition}
\newtheorem{example}[theorem]{Example}

\usepackage{cleveref}
\usepackage[inline]{enumitem}

\usepackage{xparse}
\NewDocumentCommand{\proofsubparagraph}{ m }{\textbf{#1}}

\hypersetup{
  hidelinks,
}

\usepackage[T1]{fontenc}

\usepackage[dvipsnames,svgnames]{xcolor}
\usepackage{xparse}
\usepackage{xspace}
\usepackage[skins]{tcolorbox}
\usepackage{booktabs}
\usepackage{csquotes}
\usepackage{stmaryrd}
\usepackage{mathtools}
\usepackage{makecell}
\usepackage{tikz}
\usepackage{pgfplots}
\usepackage{amssymb}
\usepackage{multirow}
\usepackage{subcaption}
\usepackage[noadjust,nomove]{cite}

\usetikzlibrary{
  automata,
  arrows.meta,
  calc,
  positioning,
  shapes.geometric,
  patterns,
  plotmarks,
}
\tikzset {
  edge with arrow/.style = {
    ->,
    >=stealth,
    shorten >=1pt,
  },
  directed/.style = {
    edge with arrow,
    node distance=2cm,
    on grid,
    semithick,
    double distance=1.5pt,
  },
  game/.style = {
    directed,
    auto,
    initial text={},
    state/.append style = {
      ellipse,
      inner sep = 1pt,
      minimum size = 10pt,
    },
    eve/.style = {
      state,
      ellipse,
    },
    adam/.style = {
      state,
      rectangle,
      inner sep = 2pt,
    },
    adamStr/.style = {
      ultra thick,
    }
  },
  lattice/.style = {
    on grid,
    semithick,
    auto,
  }
}

\pgfplotsset{
  compat=1.17,
  every x tick label/.append style = {font=\scriptsize},
  every y tick label/.append style = {font=\scriptsize},
  every axis label/.append style = {font=\scriptsize},
}

\NewDocumentCommand{\hairsp}{}{\hspace{1pt}}

\NewDocumentCommand{\ie}{}{\textsl{i.\nobreak\hairsp{}e.}\xspace}
\NewDocumentCommand{\eg}{}{\textsl{e.\nobreak\hairsp{}g.}\xspace}

\RenewDocumentCommand{\implies}{}{\Rightarrow}
\RenewDocumentCommand{\impliedby}{}{\Leftarrow}

\NewDocumentCommand{\nat}{}{\mathbb{N}}
\NewDocumentCommand{\natnonneg}{}{{\nat_{>0}}}
\NewDocumentCommand{\natsubs}{ m }{{\nat_{#1}}}
\NewDocumentCommand{\interval}{ m m }{\llbracket #1, #2 \rrbracket}
\NewDocumentCommand{\subsets}{ m }{\mathcal{P}(#1)}

\NewDocumentCommand{\complexity}{ m }{\textsf{#1}}
\NewDocumentCommand{\PTIME}{}{\complexity{PTIME}\xspace}

\NewDocumentCommand{\PSPACE}{}{\complexity{PSPACE}\xspace}

\NewDocumentCommand{\arena}{}{\mathcal{A}}
\NewDocumentCommand{\vertices}{}{V}
\NewDocumentCommand{\alphabet}{}{\Sigma}
\NewDocumentCommand{\transitions}{}{\Delta}
\NewDocumentCommand{\En}{}{\mathrm{En}}
\NewDocumentCommand{\opponents}{m O{}}{\nabla_{#2}(#1)}
\NewDocumentCommand{\histories}{}{\mathsf{Hist}}
\NewDocumentCommand{\strat}{}{\sigma}
\NewDocumentCommand{\outcome}{}{\mathsf{Out}}
\NewDocumentCommand{\target}{}{t}
\NewDocumentCommand{\game}{}{\mathcal{G}}

\NewDocumentCommand{\kGame}{}{\mathsf{Know}_{\game}}
\NewDocumentCommand{\verticesE}{}{{\vertices_E}}
\NewDocumentCommand{\verticesEK}{}{{\vertices_E^\latticeK}}
\NewDocumentCommand{\verticesA}{}{{\vertices_A}}
\NewDocumentCommand{\kTransitions}{}{{\transitions_{\mathsf{Know}}}}
\NewDocumentCommand{\final}{}{F}
\NewDocumentCommand{\kStrat}{}{\lambda}
\NewDocumentCommand{\tree}{}{\mathcal{T}}
\NewDocumentCommand{\Win}{}{\textsf{Win}\xspace}
\NewDocumentCommand{\Lose}{}{\textsf{Lose}\xspace}

\NewDocumentCommand{\parOrder}{}{\sqsubseteq}
\NewDocumentCommand{\parOrdersim}{}{\sqsubseteq_{\mathrm{sim}}}
\NewDocumentCommand{\Dom}{}{\mathsf{Dom}}
\NewDocumentCommand{\reduce}{ m }{\left\lceil #1 \right\rceil}
\NewDocumentCommand{\expand}{ m }{{#1\!\!\downarrow}}
\NewDocumentCommand{\lub}{}{\sqcup}

\NewDocumentCommand{\glb}{}{\sqcap}

\NewDocumentCommand{\W}{}{\mathsf{W}}
\NewDocumentCommand{\WAlt}{}{\mathsf{W}_{\mathrm{Alt}}}
\NewDocumentCommand{\KPred}{m m O{}}{\mathsf{KPred}_{#3}[#1](#2)}
\NewDocumentCommand{\KPredAlt}{m m
O{}}{\mathsf{KPred}^{\mathrm{Alt}}_{#3}[#1](#2)}
\NewDocumentCommand{\Pred}{m O{}}{\mathsf{Pred}_{#2}(#1)}
\NewDocumentCommand{\PredAlt}{m O{}}{\mathsf{Pred}^{\mathrm{Alt}}_{#2}(#1)}
\NewDocumentCommand{\latticeK}{}{\mathcal{K}}

\NewDocumentCommand{\bigsetminus}{}{\mathbin{\big\backslash}}

\AtEndEnvironment{example}{\hfill$\lrcorner$}

\begin{document}
\maketitle

\begin{abstract}
  Concurrent parameterized games involve a fixed yet arbitrary
  number of players.
  They are described by finite arenas in which
  the edges are labeled with languages that describe the possible
  move combinations leading from one vertex to another ($n$ players
  yield a word of length $n$).
  Previous work showed that, when edge labels are regular languages,
  one can decide whether a distinguished
  player, called Eve, has a strategy to ensure a reachability
  objective, against any strategy profile of her arbitrarily many
  opponents. This decision problem is known to be
  \PSPACE-complete. A basic ingredient in the
  \PSPACE-membership proof is
  the reduction to the exponential-size \emph{knowledge game}, a
  2-player game that reflects the knowledge Eve has on the number of
  opponents.

  In this paper, we provide a symbolic approach, based on
  antichains, to compute Eve's winning region in the knowledge
  game. In words, it gives the minimal knowledge Eve needs at every
  vertex to win the concurrent parameterized reachability game. More
  precisely, we propose two fixed-point algorithms that compute, as
  an antichain, the maximal elements of the winning region for Eve
  in the knowledge game. We implemented these two algorithms in C++,
  as well as the one initially proposed, and report on their
  relative performances on various benchmarks.
\end{abstract}

\section{Introduction}\label{sec:intro}

\paragraph{Parameterized verification.} The verification of models
with unknown parameters is known as parameterized
verification. Special cases are systems formed of an arbitrary number
of agents. Motivations for such models are varied: they range from
distributed algorithms and communication protocols to chemical or
biological systems. In all cases, these parameterized systems are
meant to be correct independently of the number of agents.  Verifying
parameterized systems is challenging, since it aims at validating
infinitely many instances at once, independently of the parameterized
number of agents.

Apt and Kozen proved the first undecidability result for the
verification of parameterized systems~\cite{AK-ipl86}. In their
reduction, the system consists of a single agent, whose description
depends on parameter $n$. Oftentimes, when agents bear identities, the
parameterized verification of simple properties is
undecidable. In contrast, in their pioneering work on distributed
models with many
identical entities, German and Sistla represented the behavior of a
network by finite state machines interacting via rendezvous
communications, and provided decision algorithms for their
parameterized verification~\cite{GS-jacm92}. Since then, variants of
models have been proposed to handle various communication
means. Remarkably, subtle changes, such as the presence or absence of
a controller in the system, can drastically alter the complexity of
the verification problems, as surveyed in~\cite{Esparza-stacs14}.

\paragraph{Concurrent parameterized games.}
Concurrent parameterized games~\cite{BertrandBM19} are one of the few
examples of game frameworks in the parameterized verification
landscape. They generalize concurrent 2-player games to an arbitrary number of
players. More precisely, the arena of a concurrent parameterized game
represents infinitely many games, each played with a fixed number of
players. A distinguished player, Eve, aims at achieving her objective
not knowing \emph{a priori} how many opponents she has. Solving such
games amounts to determining whether Eve has a winning strategy whatever
the number of her opponents and their respective strategies. When Eve
wants to achieve a reachability objective, deciding whether she has a
winning strategy, independently of the number of her opponents, is a
\PSPACE-complete problem~\cite{BertrandBM19}. Since Eve's opponents
can collude, the resolution of concurrent parameterized games with
reachability objectives reduces to the resolution of a two-player
game, called the \emph{knowledge game}, where vertices not only store
the original game vertex but also the \emph{knowledge} Eve has on the number
of her opponents. Her adversary in that 2-player game represents the
coalition of her arbitrarily many opponents in the parameterized
game. In the worst-case, the set of possible knowledge --- and hence the
size of the knowledge game --- can be exponential in the size of the
concurrent parameterized game. The DFS algorithm proposed
in~\cite{BertrandBM19} performs an on-the-fly DFS exploration
of the knowledge game.

\paragraph{Contributions.}
We exploit a monotonicity property of the knowledge
game: if Eve wins from a given knowledge set, she also wins from any
finer knowledge set. Building on this observation and the natural
partial order on sets of natural numbers (representing possible number of
opponents), we propose two new algorithms for checking the existence
of a winning strategy for Eve, based on antichains.

We first characterize the winning positions for Eve in the knowledge
game as the least fixed-point of a monotone function on the lattice of
antichains on knowledges. We turn this non-effective fixed-point
characterization into two algorithms: the first one exploits that only
finitely many knowledge sets are relevant, making it possible to
compute in a finite lattice; the second one takes an operational
viewpoint to compute iteratively the winning region. A non-trivial
contribution consists in proving that these two algorithms indeed
compute the winning region.  Importantly, our new
algorithms do more than the aforementioned DFS algorithm, as the
latter decides whether Eve has a winning strategy from a fixed
initial vertex, while we determine every position from which Eve can
win.

We implemented our two new algorithms and the DFS one proposed
in~\cite{BertrandBM19} in C++ and benchmarked their performances on
several instances of concurrent parameterized games. All instances are
scalable, to demonstrate how the three algorithms scale when instances
grow large (both in number of vertices and relevant knowledge sets).
Inspired by the \PSPACE-hardness proof of~\cite{BertrandBM19}, part of
our benchmarks are concurrent parameterized arenas derived from
quantified boolean formulas (QBF), with up to 7 variables, and a
number of clauses between 50 and 100.
Our benchmarks cover instances of concurrent parameterized games with
up to 126 vertices and hundreds of edges, resulting in a lattice of
2048 to 32768 knowledge sets (per vertex).  We were able to
successfully check the existence or the absence of a winning strategy
for Eve on instances with up to thousands of knowledge vertices in
less than 5 minutes. The experiments show that on instances
of synthetic arenas that are winning for Eve, the running times are
favorable to the
DFS-algorithm. This competitive drawback of our algorithms must be set
against the fact that our algorithms not only decide whether Eve has a
winning strategy, but also compute her winning region: they provide
for every vertex
the minimal information she needs to win.
Conversely, on QBF instances, the algorithm from~\cite{BertrandBM19}
times out for much smaller instances than our antichain
algorithms. Also, there are many negative synthetic instances of
parameterized games, on which the antichain algorithms outperform the
DFS one. Finally, among antichain algorithms, the experiments
highlight the significantly better performance of the one that avoids
the construction of the lattice.

\paragraph{Related work.} Fixed-point computations on lattices were
successfully used in the
verification community to efficiently solve computationally hard
problems that are not \PTIME-tractable: resolution of games of
imperfect information~\cite{DDR-hscc06}, universality of finite
automata~\cite{DDHR-cav06,DR-tacas10}, language inclusion for various
types of automata~\cite{DR-lmcs09,HIRV-fmsd20,DGH-tacas23},
an automata-based approach for LTL
model-checking~\cite{DDMR-tacas08,DR-lmcs09}, LTL
realizability~\cite{FJR-cav09,FJR-fmsd11,CP-tacas23}, QBF
satisfiability~\cite{BBDDR-atva11}, or refinement
checking~\cite{WSSLDWL-icfem12,LGW-lmcs21}.  An antichain is merely a
set of incomparable
elements for a given partial order.  The antichain approach always
exploits some monotonicity of the problem and proposes a fixed-point
(semi-)algorithm iterating a monotone operator over the lattice of
antichains. In these works, the greatest or least fixed-point of the
operator exactly contains the relevant information to solve the
problem. Generally speaking the antichain approach compared favorably,
in terms of experimental complexity, to other approaches from the
literature. The work closest to ours in that series of contributions
is the earliest one, on games of imperfect information. Not only is it
the only one concerned with games, but also, similarly to our setting,
players only have partial information. The two settings however differ
greatly in that in concurrent parameterized games, the information Eve
has is very specific and concerns the number of opponents she has,
whereas in games of imperfect information, the uncertainty is on the
current game vertex.
\section{Concurrent Parameterized Games}\label{sec:setting}

Concurrent parameterized games, introduced in~\cite{BertrandBM19}, are
described by parameterized arenas. In such an arena, edges are labeled
by pairs $(a,k)$, where $a$ is an action of Eve and $k$ a number of
opponents\footnote{Our definition of parameterized arenas specifies
  the number of opponents, and does not reflect their action choices.
  While parameterized arenas were initially defined with edges labeled
  with regular languages, for the problem we are interested in, one
  can focus on the number of opponents without loss of generality
(see~\cite[Section 4]{BertrandBM19}).}. The number of opponents is
initially unknown to Eve, but she may refine her knowledge about $k$
while playing, based on the alternating sequence of actions and
visited vertices. The objective of Eve is to ensure a reachability
objective independently of $k$.

We write $\natnonneg$ for the set of positive integers
and more generally $\natsubs{\gamma}$ for the subset of positive
integers satisfying a constraint $\gamma$, \eg,
$\natsubs{>2}$ for the set of integers larger than $2$.

\begin{definition}[Parameterized arena]\label{def:arena}
  A \emph{parameterized arena} is a tuple
  \(\arena = (\vertices, \alphabet, \transitions)\)
  where
  \(\vertices\) is a finite set of vertices,
  \(\alphabet\) is a finite set of actions, and
  \(\transitions : \vertices \times \alphabet \times \natnonneg \to
  \subsets{\vertices}\) is the transition function.
\end{definition}

The arena is \emph{deterministic} if for every \(v \in \vertices\),
\(a \in \alphabet\), and \(k \in \natnonneg\), there is at most one
vertex \(v' \in \vertices\) such that
\(v' \in \transitions(v, a, k)\).  An action \(a \in \alphabet\) is
\emph{enabled} at vertex \(v\) if there exists \(k \in \natnonneg\)
such that \(\transitions(v, a, k) \neq \emptyset\).  We write
\(\En(v)\) for the set of enabled actions at \(v\).  The arena is
assumed to be \emph{complete for enabled actions}: for every
\(v \in \vertices\), if \(a\) is enabled at \(v\), then for all
\(k \in \natnonneg\), \(\transitions(v, a, k) \neq \emptyset\). This
assumption is natural: Eve does not know how many opponents she has,
and the successor vertex must exist whatever that number is.

For any \(v, v' \in \vertices\) and \(a \in \alphabet\), we introduce
the following notation to represent the set of numbers of opponents that can
lead from \(v\) to \(v'\) under action \(a\) of Eve:
\(\opponents{v, a, v'} = \{k \in \natnonneg \mid v' \in
\transitions(v, a, k)\}\).

For every $v\in \vertices$, for every $a \in \En(v)$ and every
$V' \subseteq V$, we define:
\begin{equation}\label{nablaI}
  \opponents{v,a,V'} = \bigcap_{v' \in V'} \opponents{v,a,v'}
  \bigsetminus
  \bigcup_{v' \notin V'} \opponents{v,a,v'}.
\end{equation}
This corresponds exactly to the numbers of opponents that allow to
reach vertices in $V'$ and forbid to reach vertices outside $V'$.
Note that, on non-deterministic arenas, we may have
\(\opponents{v, a, \{v'\}} \subsetneq \opponents{v, a, v'}\).
As a consequence of the completeness assumption, for every
$v \in \vertices$ and $a \in \En(v)$, for every set
$K \subseteq \natnonneg$:
\begin{equation}\label{K:complete}
  K = \bigcup_{V' \subseteq V} K \cap \opponents{v,a,V'}.
\end{equation}

Although \Cref{def:arena} is quite general and permits
infinitely many outgoing transitions (one for each natural number,
for example), in the rest of the paper, we impose restrictions on
the arenas. Specifically, we assume that all sets
\(\opponents{v, a, v'}\) and \(\opponents{v, a, V'}\) are finite
unions of intervals, as in our implementation
(see~\Cref{sec:experiments}).  However, the developments made
throughout this paper also apply to any family of sets with effective
intersection, union and set difference operations,
provided it yields a finite representation of the
arena, such as semilinear sets for instance.

\begin{figure}[t]
  \centering
  \begin{tikzpicture}[
      game,
      node distance = 20pt and 70pt,
    ]
    \node [state]                               (v)   {\(v\)};
    \node [state, above right=of v]             (x1)  {\(x_1\)};
    \node [state, below right=of v]             (x2)  {\(x_2\)};
    \node [state, right=of x1]                  (y1)  {\(y_1\)};
    \node [state, right=of x2]                  (y2)  {\(y_2\)};
    \node [state, accepting, below right=of y1] (t)   {\(t\)};
    \node [state, left=of v]              (bad) {\(s\)};

    \path
    (v) edge node [sloped, pos=0.4, yshift=-3pt]  {\(a, \natnonneg\)}
    (x1)
    edge node [sloped, pos=0.6, yshift=-3pt]  {\(a, \natnonneg\)}
    (x2)
    edge node [sloped]                        {\(c,  \natsubs{\leq
    2}\)}     (bad)
    edge node                                 {\(c, \natsubs{> 2}\)}         (t)
    (x1)  edge node                                 {\(b, \{1\}\)}
    (y1)
    (x2)  edge node [']                             {\(b, \{2\}\)}
    (y2)
    (y1)  edge node                                 {\(b,
    \natnonneg\)}            (t)
    (y2)  edge node [']                             {\(b,
    \natnonneg\)}            (t)
    ;

    \draw [rounded corners = 10pt]
    let
    \p{s} = (x1.west),
    \p{t} = (v.north),
    in
    (\p{s}) -- (\x{t}, \y{s})
    node [above, pos=0.4] {\(b,  \natsubs{\neq 1}\)}
    -- (\p{t})
    ;
    \draw [rounded corners = 10pt]
    let
    \p{s} = (x2.west),
    \p{t} = (v.south),
    in
    (\p{s}) -- (\x{t}, \y{s})
    node [below, pos=0.4] {\(b,  \natsubs{\neq 2}\)}
    -- (\p{t})
    ;
  \end{tikzpicture}
  \caption{An example of a non-deterministic parameterized
  arena.}\label{fig:arena:nondeterministic}
\end{figure}
\begin{example}
  Let \(\arena\) be the parameterized arena of
  \Cref{fig:arena:nondeterministic}.  In our figures, we use
  constraints to represent the transition function: for instance, the
  label \enquote{\(c, \natsubs{\leq 2}\)} on the transition from \(v\)
  to \(s\) represents
  \(\transitions(v, c, 1) = \transitions(v, c, 2) = \{s\}\).

  The label \enquote{\(a, \natnonneg\)} from \(v\) to \(x_1\) denotes
  that \(\opponents{v, a, x_1} = \natnonneg\), while the absence of
  edges from \(v\) to \(y_1\) represents
  \(\opponents{v, a, y_1} = \emptyset\).  Further, as
  \(\opponents{v, a, x_2} = \natnonneg\) and, for any
  \(v' \in \{s, y_1, y_2, t\}\), \(\opponents{v, a, v'} = \emptyset\),
  we have \(\opponents{v, a, \{x_1, x_2\}} = \natnonneg\).  Finally,
  since \(\opponents{x_1, b, v} =  \natsubs{\neq 1}\) and
  \(\opponents{x_1, b, y_1} = \{1\}\), it holds that
  \(\opponents{x_1, b, \{v, y_1\}} = \emptyset\).
  Note that this arena
  is complete for enabled actions.
\end{example}

Let \(k \in \natnonneg\).
A \emph{\(k\)-history}, for a coalition composed of \(k\) opponents of Eve, is
a finite sequence
\(v_0 a_0 \dotsb v_i \in {(\vertices \cdot \alphabet)}^* \cdot \vertices\)
such that for every \(j < i, v_{j+1} \in \transitions(v_j, a_j, k)\).
A \emph{history} in \(\arena\) is a \(k\)-history for some \(k \in \natnonneg\).
We write \(\histories(k)\) (resp.\ \(\histories\)) for the set of
\(k\)-histories
(resp.\ histories) in \(\arena\).
Similar notions of a \emph{\(k\)-play} and a \emph{play} are defined for
infinite sequences.

\begin{definition}[Strategy]
  A \emph{strategy} for Eve in \(\arena\) is a mapping
  \(\strat : \histories \to \alphabet\) that associates to every
  history \(hv' \in \histories\) an action \(\strat(hv') \in \En(v')\).
\end{definition}

A strategy for Eve is applied with no prior information on the number
of opponents. The resulting plays however may highly
depend on it. Given a strategy \(\strat\), an initial vertex \(v\)
and \(k \in \natnonneg\) a number of opponents, we define the outcome
\(\outcome(\strat, v, k)\) as the set of plays that \(\strat\) induces
from \(v\) when Eve has exactly \(k\) opponents:
\(\outcome(\strat, v, k)\) is the set of all \(k\)-plays
\(\rho = v_0 a_0 v_1 a_1 v_2 \dotsb\) such that \(v = v_0\), and for
all \(i \geq 0\), \(\strat(v_0 a_0 \dotsb v_i) = a_i\) and
\(v_{i+1} \in \transitions(v_i, a_i, k)\).  The completeness
assumption ensures that the set \(\outcome(\strat, v, k)\) is not
empty.  Finally, \(\outcome(\strat, v)\) is the set of all possible
plays induced by \(\strat\) from \(v\):
\(\outcome(\strat, v) = \bigcup_{k \geq 1} \outcome(\strat, v, k)\).

Given a parameterized arena
\(\arena = (\vertices, \alphabet, \transitions)\), a target vertex
\(\target \in \vertices\) defines a \emph{reachability game}
\(\game = (\arena, \target)\) for Eve.  A strategy \(\strat\) for Eve
from \(v\) in $\game$ is
\emph{winning} if all plays in \(\outcome(\strat, v)\) eventually
reach \(\target\).  In this case,
\(v\) belongs to the \emph{winning region} of Eve.
On the arena of \Cref{fig:arena:nondeterministic}, in which the target
vertex is highlighted by a double circle, the winning region of Eve is
$\{t, y_1,y_2\}$, as argued in \Cref{ex:knowledgeGame} below.

For constraints $\nabla(v,a,v')$ that are finite unions of intervals or
semilinear sets, the paper~\cite{BertrandBM19} studies the existence
of a winning strategy for Eve, and shows:
\begin{theorem}[\cite{BertrandBM19}]
  Deciding the existence of a winning strategy for Eve is \PSPACE-complete.
\end{theorem}

The purpose of the current paper is to use lattices and antichains to
compute the winning region of Eve.
We first recap some results of~\cite{BertrandBM19}.

\paragraph{Knowledge game.}

The DFS algorithm of~\cite{BertrandBM19} first turns a
parameterized reachability game into a standard two-player turn-based
game~\cite[Chapter 2]{GamesOnGraphs}.

\begin{definition}[Knowledge game~\cite{BertrandBM19}]\label{def:knowledge}
  Let \(\game = ((\vertices, \alphabet, \transitions), \target)\) be a
  concurrent parameterized reachability game. The \emph{knowledge
  game} associated with \(\game\) is the two-player turn-based
  reachability game
  \(\kGame = (\verticesE \cup \verticesA, \kTransitions, \final)\)
  between Eve and Adam, such that
  \(\verticesE = \vertices \times \subsets{\natnonneg}\) and
  \(\verticesA = \verticesE \times \alphabet\) are Eve and Adam
  vertices, respectively;
  \(\final = \{(\target, K) \mid K \subseteq \natnonneg\}\) is the set
  of target vertices; and
  \(\kTransitions \subseteq (\verticesE \times \verticesA) \cup
  (\verticesA \times \verticesE)\) is the edge relation defined as follows:
  \begin{itemize}
    \item
      for all \((v, K) \in \verticesE\) and \(a \in \alphabet\)
      enabled at \(v\),
      \(((v, K), (v, K, a)) \in \kTransitions\);
    \item
      for all \((v, K, a) \in \verticesA\) and \(v' \in \vertices\),
      \(((v, K, a), (v', K \cap \opponents{v, a, v'})) \in \kTransitions\)
      whenever \(K \cap \opponents{v, a, v'} \neq \emptyset\).
  \end{itemize}
\end{definition}

Intuitively, each vertex \((v, K)\) owned by Eve records
her current knowledge on the number of her opponents.  We call the
set \(K\) a \emph{knowledge set}.  When she selects a successor
\((v, K, a)\), Adam picks a vertex
\((v', K \cap \opponents{v, a, v'})\) and the knowledge of Eve is updated:
if she learns that the number of opponents cannot be
some \(k\), she retains this information forever, and can use it when
selecting a later edge.
Notice that the reachable part
of \(\kGame\) from any \((v, \natnonneg)\) is finite: one can show that
for every vertex \((v, K)\) that Eve can reach, \(K\) is the
intersection of finitely many sets of the form
\(\opponents{v', a, v''}\), or \(\natnonneg\).

A \emph{strategy} for Eve in \(\kGame\) is a mapping
\(\kStrat : {(\verticesE \cdot \verticesA)}^* \cdot \verticesE \to \verticesA\)
compatible with \(\kTransitions\).
We borrow standard notions of outcomes and winning strategies from
the literature.

\begin{theorem}[\cite{BertrandBM19}]\label{thm:knowledge:strategy}
  Eve has a winning strategy
  from \(v\) in \(\game\) if and only if she has a winning
  strategy from \((v, \natnonneg)\) in \(\kGame\).
\end{theorem}

\begin{figure}[t]
  \centering
  \begin{tikzpicture}[
      game,
      node distance = 27pt and 70pt,
    ]
    \node [eve]                             (v N)       {\(v, \natnonneg\)};
    \node [adam, below=of v N]              (v N c)     {\(v, \natnonneg, c\)};
    \node [adam, right=of v N]              (v N a)     {\(v, \natnonneg, a\)};
    \node [eve, below=of v N c]              (s leq 2)  {\(s,
    \natsubs{\leq 2}\)};
    \node [eve, accepting, right=of v N c]  (t gt 2)    {\(t,  \natsubs{> 2}\)};
    \node [eve, right=of v N a]             (x1 N)      {\(x_1, \natnonneg\)};
    \node [eve, below=of x1 N]              (x2 N)      {\(x_2, \natnonneg\)};
    \node [adam, right=of x1 N]             (x1 N b)    {\(x_1,
    \natnonneg, b\)};
    \node [adam, right=of x2 N]             (x2 N b)    {\(x_2,
    \natnonneg, b\)};
    \node [eve, above=of x1 N]              (v neq 1)   {\(v,
    \natsubs{\neq 1}\)};
    \node [eve, below=of x2 N]              (v neq 2)   {\(v,
    \natsubs{\neq 2}\)};
    \node [right=of v neq 2]                            {\ldots};
    \node [eve, right=of x1 N b]            (y1 1)      {\(y_1, \{1\}\)};
    \node [eve, right=of x2 N b]            (y2 2)      {\(y_2, \{2\}\)};
    \node [adam, right=of y1 1]             (y1 1 b)    {\(y_1, \{1\}, b\)};
    \node [adam, right=of y2 2]             (y2 2 b)    {\(y_2, \{2\}, b\)};
    \node [eve, accepting, above=of y1 1 b] (t 1)       {\(t, \{1\}\)};
    \node [eve, accepting, below=of y2 2 b] (t 2)       {\(t, \{2\}\)};
    \node [adam, right=of v neq 1]          (v neq 1 c) {\(v,
    \natsubs{\neq 1}, c\)};
    \node [right=of v neq 1 c]                          {\ldots};
    \node [adam, above=of v neq 1]          (v neq 1 a) {\(v,
    \natsubs{\neq 1}, a\)};
    \node [eve, left=of v neq 1 a]          (x1 neq 1)  {\(x_1,
    \natsubs{\neq 1}\)};
    \node [eve, right=of v neq 1 a]         (x2 neq 1)  {\(x_2,
    \natsubs{\neq 1}\)};
    \node [right=of x2 neq 1]                           {\ldots};
    \node [adam, left=of x1 neq 1]          (x1 neq 1 a){\(x_1,
    \natsubs{\neq 1}, b\)};

    \foreach \s/\t in {v N/v N a,v N/v N c,
      v N c/t gt 2,
      v N a/x2 N,
      x1 N/x1 N b,
      x2 N/x2 N b,
      x1 N b/y1 1,
      x2 N b/y2 2,
      y1 1/y1 1 b,
      y2 2/y2 2 b,
      v neq 1/v neq 1 c,v neq 1/v neq 1 a,
      v neq 1 a/x2 neq 1,
    x1 neq 1/x1 neq 1 a} {
      \path (\s) edge (\t);
    }

    \foreach \s/\t in {v N a/x1 N,
      x1 N b/v neq 1,
      v neq 1 a/x1 neq 1,
      x2 N b/v neq 2,
      y1 1 b/t 1,
      y2 2 b/t 2,
    v N c/s leq 2} {
      \path (\s) edge [adamStr] (\t);
    }

    \draw [rounded corners = 10pt, adamStr]
    let
    \p{s} = (x1 neq 1 a.south),
    \p{t} = (v neq 1.west),
    in
    (\p{s}) -- (\x{s}, \y{t})
    -- (\p{t})
    ;
  \end{tikzpicture}
  \caption{Fragment of the knowledge game from \(v\) for the arena of
  \Cref{fig:arena:nondeterministic}.}\label{fig:knowledgeGame}
\end{figure}
\begin{example}\label{ex:knowledgeGame}
  \Cref{fig:knowledgeGame} represents a fragment of the knowledge
  game associated
  with the parameterized arena of \Cref{fig:arena:nondeterministic}.
  Oval vertices belong to Eve, and rectangular ones to Adam.
  In this two-player game, Eve does not have a winning strategy
  from \((v, \natnonneg)\).
  A memoryless counterstrategy for Adam is drawn in bold in the
  figure. In particular, if Eve chooses to go to
  \((v, \natnonneg, a)\), Adam's strategy forces a cycle in the game
  that never visits \(t\).
  However, for any \(K \subseteq \natnonneg\), Eve has a winning
  strategy from \((y_1, K)\) and from \((y_2, K)\). In particular, she
  wins from \(y_1\) and from \(y_2\) in the parameterized game, thanks
  to \Cref{thm:knowledge:strategy}.
\end{example}

\section{Symbolic algorithms for parameterized games}\label{sec:symbolic}

Let \(\game\) be a parameterized reachability game, and \(\kGame\) be
its associated knowledge game.  Our objective is to compute Eve's
winning region in \(\kGame\).  Thanks to
\Cref{thm:knowledge:strategy}, deriving Eve's winning region in
\(\game\) will then be straightforward.  Instead of storing every
vertex of the winning region in \(\kGame\), we obtain a more compact
representation by leveraging \emph{antichains}.  We first formally
define antichains and their properties, before specializing them to
represent vertices in \(\kGame\).  We then give a fixed-point
characterization of the winning region via antichains and show its
correctness, before providing a different, more operational
fixed-point approach. We finally prove the non-trivial
fact that both algorithms compute the same sets, \ie, they both
output Eve's winning region in \(\kGame\).

\subsection{A summary of the theory of antichains}\label{subsec:antichains}

Before introducing antichains (in short, sets of incomparable elements according
to some partial order), let us focus on partially ordered sets
and (complete) lattices.

\subparagraph{Partially ordered sets.}

Formally, for \(X\) a set, a \emph{partial order} on $X$
is a binary relation \({\parOrder} \subseteq X \times X\) that is
reflexive, transitive, and antisymmetric (\ie, for any \(a, b \in X\),
if \(a \parOrder b\) and \(b \parOrder a\), then \(a = b\)).  The pair
\((X, \parOrder)\) denotes a \emph{partially ordered set}.
We say that \((X, \parOrder)\) is a \emph{lattice} if every
pair of elements \(x \neq y\) of \(X\) has a
\emph{join} (a least upper bound) and a \emph{meet} (a greatest lower bound)
that both belong to \(X\).

\begin{example}\label{ex:lattice}
  The partially ordered set \((\subsets{\nat}, \subseteq)\) is
  a lattice, with the join defined as the union, and
  the meet as the intersection. For instance, the join of \(\{1, 3\}\)
  and \(\{2, 3\}\) is \(\{1, 2, 3\}\), while their meet is
  \(\{3\}\). As both results belong to \(\subsets{\nat}\) (and this
  can be easily generalized to any pair of sets), \((\subsets{\nat},
  \subseteq)\) is indeed a lattice.

  For another example, let
  us consider the partially ordered set \((\{1, 2, 3, 6\},
  \parOrder)\) such that,
  for any \(a, b \in \nat\), \(a \parOrder b\) if and only if \(a\)
  divides \(b\).
  Then, by defining the join as the least common multiple (\eg, the
  join of 2 and 3 is 6) and the meet as the greatest common divisor
  (\eg, the meet of 2 and 3 is 1), \((\{1, 2, 3, 6\}, \parOrder)\)
  forms a lattice. However, the partially ordered set \((\{1, 2, 3\},
  \parOrder)\)
  does not form a lattice, as the join of 2 and 3 is not defined within
  \(\{1, 2, 3\}\).
\end{example}

\subparagraph{Complete lattices.}
Throughout this paper, we will rely on special lattices: a partially ordered
set \((X, \parOrder)\) is a \emph{complete lattice} if, for every
subset \(L \subseteq X\), \(L\) has both a join and a meet in
\((X, \parOrder)\). Notice that any complete lattice is necessarily a lattice.

\begin{example}\label{ex:lattice:complete}
  The lattice \((\subsets{\nat}, \subseteq)\) is actually a complete lattice.
  Indeed, for any \(L \subseteq \subsets{\nat}\), let join be the union of
  \(L\) and meet be the intersection.
  Then, the join of \(\{\{1, 3\}, \{2, 3\}, \{3, 4\}\}\) is \(\{1, 2, 3, 4\}\),
  and the meet is \(\{3\}\).

  The lattice \((\{1, 2, 3, 6\}, \parOrder)\) from the previous example is also
  a complete lattice. Indeed, since there are finitely many elements in this
  lattice, any subset must also be finite and, thus, has clearly defined
  join and meet. This holds for any non-empty finite lattice.
  Naturally, since \((\{1, 2, 3\}, \parOrder)\) is not
  a lattice, it cannot be a complete lattice.
\end{example}

\begin{proposition}\label{prop:lattice:finite:complete}
  Any non-empty finite lattice is a complete lattice.
\end{proposition}

\subparagraph{Antichains.}

Our presentation is inspired from~\cite{DDR-hscc06}.  In
short, antichains are sets of incomparable elements of a partially
ordered set \((X, \parOrder)\).
An \emph{antichain} for \((X, \parOrder)\) is a subset \(X' \subseteq X\)
such that any two elements of \(X'\) are incomparable for
\(\parOrder\): \(\forall a \neq b \in X'\), \(a \not\parOrder b\) and
\(b \not\parOrder a\).

Let \(L \subseteq X\) be a set. An antichain can be used to represent or
abstract \(L\) by keeping only the maximal elements for \(\parOrder\).
We say that an element \(x \in X\) is \emph{dominated} by \(L\) if
there exists some \(x' \in L\) such that \(x \parOrder x'\) and
\(x \neq x'\).  We write \(\Dom(L)\) for the set of all elements of
\(X\) that are dominated by \(L\):
\(\Dom(L) = \{x \in X \mid \exists x' \in L : x \parOrder x' \land x
\neq x'\}\).

The \emph{reduced form} of \(L\), noted \(\reduce{L}\), is obtained by
removing from \(L\) all its dominated elements, while the
\emph{expanded form} of \(L\), noted \(\expand{L}\), is obtained by
adding to $L$ all dominated elements in $X$.  That is,
\( \reduce{L} = L \setminus \Dom(L) \) and
\( \expand{L} = L \cup \Dom(L)\).
As \(\Dom(L) \cap \reduce{L} = \emptyset\) (\ie, \(\reduce{L}\) does not contain
any element that is dominated by \(L\)), any two distinct elements
of \(\reduce{L}\) are incomparable for \(\parOrder\), and
\(\reduce{L}\) is an antichain.
Notice that, for any \emph{finite} non-empty subset \(L\) of \(X\), there is
at least one element that is not dominated by \(L\). That is, \(L \nsubseteq
\Dom(L)\), meaning that \(\reduce{L}\) is not empty.
Moreover, \(\expand{L} = \expand{\reduce{L}}\). This
does not always hold for infinite sets: for instance, the subset \(L =
\{\{1\}, \{1, 2\}, \dotsc\}\) of \((\subsets{\nat}, \subseteq)\) does
not have a maximal element, \ie, \(\reduce{L} =
\emptyset\). Among other results, the next proposition gives a
necessary and sufficient condition for \(\expand{L} =
\expand{\reduce{L}}\) to hold, even for infinite sets.

\begin{restatable}{proposition}{propExpandReduce}\label{prop:expand_is_expand_reduce}
  For any sets \(L, L' \subseteq X\), the following facts all hold.
  \begin{enumerate}
    \item
      \(\reduce{L}\) is an antichain for
      \((X, \parOrder)\).
    \item
      \(\Dom(L) = \Dom(\expand{L})\).
    \item
      \(\expand{L} \cup \expand{L'} = \expand{(L \cup L')}\).
    \item
      \(\reduce{L} = \reduce{\expand{L}}\).
    \item
      \(\expand{L} = \expand{\reduce{L}} \iff \forall x\in L, \exists
      y \in \reduce{L} : x \parOrder y\).
    \item
      \(\expand{L} = \expand{\reduce{L}} \land
        \expand{L'} = \expand{\reduce{L'}}
        \implies
        \expand{\reduce{L \cup L'}} = \expand{(\reduce{L} \cup
      \reduce{L'})}\).
  \end{enumerate}
\end{restatable}

A set $L$ such that $L = \expand{L}$ is said \emph{downward-closed}.
Antichains are particularly adapted to representing
downward-closed sets. Indeed, if $L$ is downward-closed,
$\expand{\reduce{L}} = L$, so that the set $L$ can be uniquely
determined from the antichain $\reduce{L}$.

\subparagraph{Join and meet for antichains.}
Let us argue that the set of antichains of \((X, \parOrder)\) is a
complete lattice.
To that end, we define the join and meet of two sets \(L, L'
\subseteq X\) as follows.
\begin{description}
  \item[Join] An element \(x\) in
    \(\expand{L} \cup \expand{L'}\) is maximal when it is maximal in
    one of the two sets \(L\) or \(L'\). Hence, we can define the
    \emph{join} operation of \(L\) and \(L'\) as
    \(L \lub L' = \reduce{\expand{L} \cup \expand{L'}}\).
    Since \(\expand{L} \cup \expand{L'} = \expand{(L \cup L')}\),
    \(L \lub L' = \reduce{\expand{L} \cup \expand{L'}} =
    \reduce{\expand{(L \cup L')}} = \reduce{L \cup L'}\), by
    \Cref{prop:expand_is_expand_reduce}.
  \item[Meet] The \emph{meet}
    operation of \(L\) and \(L'\) is
    \(L \glb L' = \reduce{\expand{L} \cap \expand{L'}}\). It
    does not coincide in general with $\reduce{L \cap L'}$.
\end{description}

We emphasize that \(L \lub L'\) and \(L \glb L'\) are antichains
for \((X, \parOrder)\).  The join operation provides the least upper
bound, while the meet operation gives the greatest lower bound
for the two sets \(L\) and \(L'\).
Thus, the set of all antichains of \((X, \parOrder)\) forms a
\emph{complete lattice}.
When \(\expand{L} = \expand{\reduce{L}}\) and
\(\expand{L'} = \expand{\reduce{L'}}\), one can show that computing
\(L \lub L'\) is equivalent to computing the join of the maximal
elements of \(L\) and \(L'\).

\begin{proposition}\label{prop:lub:reduceLub}
  Let \((X, \parOrder)\) be a partially ordered set, and
  \(L, L' \subseteq X\) be two sets such that \(\expand{L} =
  \expand{\reduce{L}}\) and \(\expand{L'} = \expand{\reduce{L'}}\).
  Then,
  \(L \lub L' = \reduce{L} \lub \reduce{L'}\).
\end{proposition}

Finally, we lift $\parOrder$ to subsets of $X$: for two
subsets $L,L' \subseteq X$, we write $L \parOrdersim L'$ whenever for
every $a \in L$, there is $a' \in L'$ such that $a \parOrder a'$. In
particular,
$L \subseteq L'$ implies $L \parOrdersim L'$,
$\reduce{L} \parOrdersim L \parOrdersim \expand{L}$, and
$\expand{L} \parOrdersim L \parOrdersim \reduce{L}$.

The following lemma, whose proof is in \Cref{app:proofs:antichain},
relates several of the above-defined notions, and will be useful in
the sequel:
\begin{restatable}{lemma}{lemmareduce}\label{lemma:reduce}
  For any partially ordered set \((X, \parOrder)\)
  and \(L, L' \subseteq X\) such that \(\expand{L} =
  \expand{\reduce{L}}\) and \(\expand{L'} = \expand{\reduce{L'}}\), the three
  following properties are
  equivalent:\[
    \mathit{1.}\ \reduce{L} = \reduce{L'}; \qquad \mathit{2.}\ L
    \parOrdersim L'\ \text{and}\ L'
    \parOrdersim L; \qquad \mathit{3.}\ L \subseteq
    \expand{L'}\ \text{and}\ L' \subseteq \expand{L}.
  \]
\end{restatable}

\subsection{Lattices and antichains for the vertices of
Eve}\label{subsec:appli-antichains}

Aiming at using antichains to compute the winning region for Eve in
\(\kGame\), we make the following observation: if Eve has a winning
strategy from some \((v, K) \in \verticesE\), then she also wins from
any \((v, K')\) such that \(K' \subseteq K\).  In words, refining the
knowledge on the number of opponents can only help Eve to reach a
target vertex. Therefore, her winning region is downward-closed for a
natural partial order on her vertices in $\kGame$, which we now define.

Given a set \(S\), we write \(\subsets{S}\) for the set of all subsets
of \(S\).  Now, one naturally defines the following partial order on
vertices of Eve in the knowledge game: for any
\((v, K), (v', K') \in \verticesE\), \((v, K) \parOrder (v', K')\) if
\(v = v'\) and \(K \subseteq K'\). Equipped with this partial order,
$(\verticesE,\parOrder)$ is a complete lattice (as
  \((\subsets{\natnonneg}, \subseteq)\) is a complete
lattice), which will be
used to compute the winning region for Eve in \(\kGame\).

We use the notation \(\expand{\cdot}\) for downward closure both in
\(\subsets{\natnonneg}\) and \(\verticesE\) (no
  confusion should arise since elements of these two lattices are
``typed'' differently). Recall that $\reduce{\cdot}$
denotes the maximal elements.

\begin{example}
  Let \(\arena\) be the arena of \Cref{fig:arena:nondeterministic} and
  \(L\) be
  \(\{(t, \natnonneg), (y_1, \{1\}),\allowbreak (y_1, \{1, 2\}),
  (y_2, \{1, 2\})\}\).
  The pair \((y_1, \{1\})\) is dominated by \(L\), as \((y_1, \{1, 2\}) \in L\),
  \((y_1, \{1\}) \parOrder (y_1, \{1, 2\})\), and
  \((y_1, \{1\}) \neq (y_1, \{1, 2\})\).
  Moreover, any \((t, K)\) with \(K \subsetneq \natnonneg\) is dominated by
  \(L\).
  However, \((y_1, \{1, 2\})\) is not dominated, as the only element \((v, K)\)
  of \(L\) such that \((y_1, \{1, 2\}) \parOrder (v, K)\) is equal to
  \((y_1, \{1, 2\})\).
  Likewise, \((t, \natnonneg)\) and \((y_2, \{1, 2\})\) are not dominated.
  Thus,
  \begin{equation*}
    \Dom(L) = \{(t, K) \mid K \subsetneq \natnonneg\}
    \cup
    \{(y_1, \emptyset), (y_1, \{1\}), (y_1, \{2\}),
    (y_2, \emptyset), (y_2, \{1\}), (y_2, \{2\})\},
  \end{equation*}
  and \(\reduce{L} = L \setminus \Dom(L) =
  \{(t, \natnonneg), (y_1, \{1, 2\}), (y_2, \{1, 2\})\}\).
\end{example}

\subsection{Fixed-point characterization of the winning
region}\label{subsec:symbolic-algo}
Similarly to the case of standard two-player turn-based reachability
games, the winning region of Eve in the knowledge game
is the least fixed-point of an analogue of the
constrained-predecessor
operator.  Since
\(\final = \{(\target, K) \mid K \subseteq \natnonneg\}\), the
fixed-point is initialized with \((\target, \natnonneg)\), as well as
all \((v, \emptyset)\), for every
\(v \in \vertices \setminus \{\target\}\).
Indeed, \((v, \emptyset)\) represents Eve's current
knowledge on the numbers of opponents she can win against from \(v\):
a priori, she always loses.
The operator \(\Pred{\cdot}\) then reflects the edges in \(\kGame\) by
incorporating all pairs \((v, K)\) such that for some action $a$
enabled at $v$, \((v', K \cap \opponents{v, a, v'})\) is already known
to be winning for every successor vertex $v'$.

\begin{definition}\label{def:symbolic}\label{def:fixpoint}
  For $(v, a) \in \vertices \times \alphabet$ and \(S \subseteq
  \verticesE\), let
  \(\KPred{v,a}{S} \) be the operator on $\subsets{\natnonneg}$,
  and let \(\Pred{S}\) be the operator on \(\subsets{\verticesE}\)
  defined as
  \begin{gather*}
    \KPred{v,a}{S} = \{K \in \subsets{\natnonneg} \mid \forall v' \in
    \vertices : (v', K \cap \opponents{v, a, v'}) \in S\}
    \\
    \Pred{S} = \bigcup_{v\in \vertices} \bigcup_{a\in \En(v)}\{(v,K)
    \mid K \in \KPred{v,a}{S}\}.
  \end{gather*}
  Finally, the sets \({(\W^i)}_{i \in \nat} \subseteq \verticesE\) are
  defined inductively:
  \begin{align*}
    \W^0 &= \{(\target, \natnonneg)\} \cup
    \{(v, \emptyset) \mid v \in \vertices \setminus \{\target\}\}
    \\
    \forall i \geq 0 : \W^{i+1} &= \W^i \lub
    \Pred{\expand{\W^i}} \quad\quad
    \text{\small (computation made in
    lattice $(\verticesE,\parOrder)$)}.
  \end{align*}
\end{definition}
The sets ${(\W^i)}_{i \in \nat}$ are well-defined as subsets of
$\verticesE$. By definition of the join operation,
we only keep the maximal elements with respect to \(\parOrder\):
each \(\W^i\) is an antichain of \((\verticesE, \parOrder)\).
Furthermore, we only consider actions \(a\) that are enabled at \(v\)
to ensure that there exists \(v'\) such that \(\opponents{v,a,v'}\) is
not empty.
Interestingly, there is no need to construct \(\kGame\) to define the
iterations of \(\W^i\).

\begin{example}\label{ex:WK}
  Let \(\game = (\arena, t)\), with
  \(\arena\) the arena of \Cref{fig:arena:nondeterministic}.
  We illustrate the approach with the full lattice \(\subsets{\natnonneg}\).
  The initial set is
  \[
    \W^0 = \{(s, \emptyset), (v, \emptyset), (x_1, \emptyset),
      (x_2, \emptyset),
    (y_1, \emptyset), (y_2, \emptyset), (t, \natnonneg)\}.
  \]
  Now, \(\KPred{v, c}{\expand{\W^0}}\) contains \(\natsubs{>2}\), as
  \((t, \natsubs{>2}) \in \expand{\W^0}\) and, for every \(v' \neq t\),
  \((v', \emptyset) \in \expand{\W^0}\).  Similarly,
  \(\KPred{v, c}{\expand{\W^0}}\) contains
  \(\{3\}, \{4\}, \dotsc, \{3, 4\},
  \dotsc, \allowbreak{\natsubs{>3}, \natsubs{>4},} \dotsc\)
  which are all subsets of \(\natsubs{>2}\).  Recall that
  \(\W^0 \sqcup \Pred{\expand{\W^0}}\) only contains the non-dominated
  elements, \ie, only \(\natsubs{>2}\) in this case.  Moreover,
  \(\KPred{y_1, b}{\expand{\W^0}} = \KPred{y_2, b}{\expand{\W^0}}\)
  contains \(\natnonneg\) and, thus, every subset of \(\natnonneg\).
  Yet, we have
  \(\KPred{x_1, b}{\expand{\W^0}} = \KPred{x_2, b}{\expand{\W^0}} =
  \{\emptyset\}\), as \((v, \emptyset), (y_1, \emptyset)\), and
  \((y_2, \emptyset)\) all belong to \(\W^0\). So,
  \[
    \W^1 = \{(s, \emptyset), (v, \natsubs{>2}), (x_1,
      \emptyset), (x_2, \emptyset), (y_1, \natnonneg), (y_2,
    \natnonneg), (t, \natnonneg)\}.
  \]
  Observe that, at each iteration, it is sufficient to consider only
  the predecessors of the vertices whose knowledge set changed, \ie,
  only \(x_1\) and \(x_2\) in this example.  We focus on \(x_1\), as
  the computations for
  \(x_2\) are similar:
  \(\KPred{x_1, b}{\expand{\W^1}}\) does not contain \(\natnonneg\),
  as \((v, \natsubs{\neq 1}) \notin \expand{\W^1}\).  However,
  \(\natsubs{\neq 2} \in \KPred{x_1, b}{\expand{\W^1}}\), as
  \((v, \natsubs{\neq 2} \cap \natsubs{\neq 1}) = (v,
  \natsubs{>2}) \in \expand{\W^1}\) and
  \((y_1, \natsubs{\neq 2} \cap \{1\}) = (y_1, \{1\}) \in
  \expand{\W^1}\).  Hence,
  \[
    \W^2 = \{(s, \emptyset), (v, \natsubs{>2}), (x_1, \natsubs{\neq
      2}), (x_2, \natsubs{\neq 1} ), (y_1, \natnonneg), (y_2,
    \natnonneg), (t, \natnonneg)\}.
  \]
  As only the knowledge sets of \(x_1\) and \(x_2\) are different, it
  suffices to compute \(\KPred{v, a}{\expand{\W^2}}\). Doing so, one
  gets \(\W^3 = \W^2\), hence
  \({(\expand{\W^i})}_{i \in \natnonneg}\) stabilizes.
\end{example}

Observe that
\(\reduce{\W^i} = \reduce{\W^i \lub \Pred{\expand{\W^i}}}
= \W^i \lub \Pred{\expand{\W^i}} = \W^i\), as \(\lub\) already
keeps only the maximal elements.
By \Cref{prop:expand_is_expand_reduce}, \(\W^i = \reduce{\W^i} =
\reduce{\expand{\W^i}}\).
We argue that the sequence
\({(\expand{\W^i})}_{i \in \nat}\) is non-decreasing. To that end, we
show the following lemma, allowing us to apply
\Cref{prop:lub:reduceLub}.

\begin{lemma}\label{lemma:W:welldefined}
  For any \(i \in \nat\), \(\expand{\Pred{\expand{\W^i}}} =
  \expand{\reduce{\Pred{\expand{\W^i}}}}\).
\end{lemma}
\begin{proof}
  By \Cref{prop:expand_is_expand_reduce}, we show that, for any
  \((v, K) \in \Pred{\expand{\W^i}}\), there exists
  \((v', K') \in \reduce{\Pred{\expand{\W^i}}}\) such that
  \((v, K) \parOrder (v', K')\).
  By definition, it is only possible when \(v = v'\). We can thus rephrase
  it as:
  \begin{equation}
    \forall (v, K) \in \Pred{\expand{\W^i}},
    \exists (v, K') \in \reduce{\Pred{\expand{\W^i}}} :
    K \subseteq K'.
    \label{eq:lemma:W:welldefined:goal}
  \end{equation}

  Let \((v, K_1) \in \Pred{\expand{\W^i}}\). If \((v, K_1) \in
  \reduce{\Pred{\expand{\W^i}}}\), we immediately get the result,
  with \(K' = K_1\).
  Assume \((v, K_1) \notin \reduce{\Pred{\expand{\W^i}}}\), \ie,
  \((v, K_1) \in \Dom(\Pred{\expand{\W^i}})\).
  So, there exists some \((v, K_2) \in
  \Pred{\expand{\W^i}}\) such that \(K_1 \subsetneq K_2\). Again, if
  \((v, K_2)\) is a maximal element, we are done. Otherwise, there is
  some \((v, K_3) \in \Pred{\expand{\W^i}}\) such that \(K_2
  \subsetneq K_3\). We can repeat this
  again and again. Assume we never find a maximal element:
  we have the infinite sequence
  \(K_1 \subsetneq K_2 \subsetneq \dotsb\) such that each \((v, K_j) \in
  \Pred{\expand{\W^i}} \setminus \reduce{\Pred{\expand{\W^i}}}\).

  By definition of \(\Pred{\expand{\W^i}}\),
  for each \(j \in \natnonneg\), there exists \(a_j \in \En(v)\) such that,
  for any \(v' \in \vertices\), \((v', K_j \cap \opponents{v, a_j,
  v'}) \in \expand{\W^i}\).
  As \(\En(v)\) is finite, there are infinitely many \(\ell < \ell'\) such that
  \(a_\ell = a_{\ell'}\) and \(K_\ell \subsetneq K_{\ell'}\).
  So, we can group together all \(K_j\) that depend on
  the same action.
  Without loss of generality, we can assume that the
  sequence of \(K_j\) for a common action is maximal (with regard
  to \(\subseteq\)), \ie, assume that
  \(\exists a \in \En(v),
    \forall j \in \natnonneg,
    \forall v' \in \vertices :
  (v', K_j \cap \opponents{v, a, v'}) \in \expand{\W^i}\), which is
  equivalent to
  \[
    \exists a \in \En(v),
    \forall v' \in \vertices,
    \forall j \in \natnonneg :
    (v', K_j \cap \opponents{v, a, v'}) \in \expand{\W^i}.
  \]

  Let \(v' \in \vertices\).
  As \((\subsets{\natnonneg}, \subseteq)\) is a complete lattice,
  there exists \(\mathfrak{K}_a\) that is the join of the
  maximal set
  \(\{K_1, K_2, \dotsc\}\), \ie, such that, for every \(j\),
  \(K_j \subseteq \mathfrak{K}_a\). For this complete
  lattice, the join is the union of the sets (see
  \Cref{ex:lattice,ex:lattice:complete}):
  \(\mathfrak{K}_a = \bigcup_{j \in \natnonneg} K_j\).
  Notice that picking a different action \(a'\) may yield
  a different \(\mathfrak{K}_{a'}\), as the choice of the action may
  lead to different knowledge sets. Observe that, since \(\{K_1,
  K_2 \dotsc\}\) is the largest possible set, \(\mathfrak{K}_a\) is
  the maximal set for the selected action \(a\). We argue that
  there is an action that gives us \(\mathfrak{K}_a \in
  \reduce{\Pred{\expand{\W^i}}}\).

  Observe that \(K_1 \cap \opponents{v, a,
  v'} \subseteq K_2 \cap \opponents{v, a, v'} \subseteq \dotsb\), as
  \(K_1 \subsetneq K_2 \subsetneq \dotsb\).
  Recall that \(\W^i = \reduce{\W^i} = \reduce{\expand{\W^i}}\), so,
  there exists \((v, \mathfrak{K}'_a) \in \W^i\) such that, for all \(j\),
  \((v', K_j \cap \opponents{v, a, v'}) \parOrder (v, \mathfrak{K}_a)\), \ie,
  \(\mathfrak{K}'_a\) is the join of the set
  \(\{K_1 \cap \opponents{v, a, v'}, \dotsc\}\).
  Again, for
  \((\subsets{\natnonneg}, \subseteq)\), the join is the union:
  \(\mathfrak{K}'_a = \bigcup_{j \in \natnonneg} K_j \cap
  \opponents{v, a, v'}\).

  By definitions, \(\mathfrak{K}_a \cap \opponents{v, a, v'} = \bigcup_{j \in
  \natnonneg} K_j \cap \opponents{v, a, v'} = \mathfrak{K}'_a\). Hence,
  \((v', \mathfrak{K}_a \cap \opponents{v, a, v'}) \in \expand{\W^i}\),
  \ie, \((v, \mathfrak{K}_a) \in \Pred{\expand{\W^i}}\).
  It may happen that the selected \(a\) does not yield a
  maximal element: indeed, another action \(a'\) may lead us to \((v,
  \mathfrak{K}_{a'}) \in \Pred{\expand{\W^i}}\) (using exactly the
  same arguments), but such that \(\mathfrak{K}_a \subseteq
  \mathfrak{K}_{a'}\). Since there are finitely many actions, there
  necessarily exists an action that yields a maximal element of
  \(\Pred{\expand{\W^i}}\), \ie, there exists \(\mathfrak{a} \in
  \alphabet\) such that \((v, \mathfrak{K}_\mathfrak{a}) \in
  \reduce{\Pred{\expand{\W^i}}}\).
\end{proof}

\begin{corollary}\label{cor:W:WNonDecreasing}
  For every \(i \in \nat\), \(\expand{\W^i} \subseteq
  \expand{\W^{i+1}}\).
\end{corollary}
\begin{proof}
  Recall that \(W^i = \reduce{\W^i}\). So, \(\expand{\W^i} =
  \expand{\reduce{\W^i}}\).
  By \Cref{lemma:W:welldefined}, \(\expand{\Pred{\expand{\W^i}}} =
  \expand{\reduce{\Pred{\expand{\W^i}}}}\).
  Hence,
  \begin{align*}
    \expand{\W^{i+1}}
    &= \expand{(\W^i \lub \Pred{\expand{\W^i}})}
    \\
    &= \expand{\reduce{\W^i \cup \Pred{\expand{\W^i}}}}
    & \text{by definition of \(\lub\)}
    \\
    &= \expand{(\reduce{\W^i} \cup \reduce{\Pred{\expand{\W^i}}})}
    & \text{by \Cref{prop:expand_is_expand_reduce}}
    \\
    &= \expand{\reduce{\W^i}} \cup \expand{\reduce{\Pred{\expand{\W^i}}}}
    & \text{by \Cref{prop:expand_is_expand_reduce}}
    \\
    &= \expand{\W^i} \cup \expand{\reduce{\Pred{\expand{\W^i}}}}
    & \text{as \(\expand{\W^i} = \expand{\reduce{\W^i}}\).}
  \end{align*}
  From there, it is clear that \(\expand{\W^i} \subseteq \expand{\W^{i+1}}\).
\end{proof}

From that result, the existence of \(\lim_{i \to +\infty}
\expand{\W^i} = \bigcup_{i \in \nat} \expand{\W^i}\) is immediate.
Nevertheless, it is unclear whether
$\lim_{i \to +\infty} \W^i$, denoted \(\W^\infty\) hereafter, exists.
To show that the sequence stabilizes, we will require
an alternative approach, which will be proved to terminate. We will
also argue that both compute the same sets at each iteration. Hence,
\(\W^\infty\) exists.
Nonetheless, we argue that the least
fixed-point of the \(\Pred{\cdot}\) operator of \Cref{def:fixpoint}
corresponds to the winning region for Eve in the knowledge game.

\begin{restatable}{theorem}{fixpointcorrectness}\label{thm:symbolic:strategy:region}
  Let \((v, K) \in \verticesE\) such that \(K \neq \emptyset\).  Eve
  has a winning strategy from \((v, K)\) in \(\kGame\) if and only if
  there exists $i \in \natnonneg$ such that
  \((v, K) \in \expand{\W^i}\).
\end{restatable}

\begin{proof}[Sketch]
  The proof follows the same line as the correctness proof of the
  fixed-point computation of the winning region in standard two-player
  turn-based reachability games.
  From every \((v, K) \in \expand{\W^i}\), one can construct a winning
  strategy by induction on the least index \(j\) such that
  \((v, K) \in \expand{\W^j}\): select an action such
  that any successor \((v', K')\) satisfies
  \((v', K') \in \expand{\W^{j-1}}\). Conversely, if Eve has a winning
  strategy from \((v, K)\), she can ensure to progress towards a
  target vertex.  Since any target vertex necessarily belongs to the
  winning region, we can show that \((v, K)\) is also in the winning
  region.
\end{proof}

From \Cref{thm:knowledge:strategy,thm:symbolic:strategy:region}, we
conclude that
\(\bigcup_{i \in \natnonneg} \expand{\W^i}\) characterizes the winning
region for Eve in the concurrent parameterized game \(\game\).

\begin{restatable}{corollary}{fixpointcorollary}
  Eve has a winning strategy from a vertex \(v_0\) in \(\game\) if and only if
  there exists $i \in \nat$ such that \((v_0, \natnonneg) \in
  \W^i\). Furthermore, if $\W^\infty$ is defined, then Eve has a
  winning strategy from \(v_0\) in \(\game\) if and only if
  \((v_0, \natnonneg) \in \W^\infty\).
\end{restatable}

\subsection{An alternative fixed-point approach}\label{subsec:alternative}
The previous fixed-point characterization cannot be immediately turned
into an algorithm, because, even if the parameterized arena is finite,
the operator $\KPred{\cdot,\cdot}{\cdot}$ collects a priori
unboundedly many sets $K$. Therefore, it is not obvious how to use
it to compute the winning region in the knowledge game. We thus
propose an alternative approach, with an operational viewpoint, to
compute \(\W^{i+1}\) from \(\W^i\).  Let \(v \in \vertices\) and
\(a \in \En(v)\).  Before providing the definition, we give the
intuition behind it:
\begin{itemize}
  \item
    Fix, for each \(v' \in V\), a knowledge set \(K_{v'}\) such that
    \((v', K_{v'}) \in \W^i\).
  \item We want to focus on the vertices \(v'\) that can be reached from
    \(v\) by taking an edge labeled with \(a\) and some constraint on
    the number of opponents.  Since the arena may be non-deterministic,
    there can be multiple such vertices.  Hence, we iterate over each
    alternative, by taking every possible subset \(V'\) of \(V\),
    representing the vertices that we want to reach.

  \item We first compute the intersection of every \(K_{v'}\), with
    \(v' \in V'\).  This yields a set of numbers of opponents for which
    we are sure to have a winning strategy from each \(v'\).  Indeed,
    for every \(v' \in V'\), \((v', K_{v'}) \in \W^i\).  As
    \(\bigcap_{v' \in V'} K_{v'}\) is a subset of each \(K_{v'}\), we have
    \((v', \bigcap_{v' \in V'} K_{v'}) \in \expand{\W^i}\).

  \item It remains to check that we can guarantee to \emph{only} reach
    vertices in \( V'\), while avoiding any vertex \(v''\) with
    \(v'' \notin V'\).  That is, we want to reach
    $\opponents{v,a,V'}$, as defined in~\Cref{nablaI}
    page~\pageref{nablaI}.

  \item We take the intersection of
    \(\bigcap_{v' \in V'} K_{v'}\) and \(\opponents{v,a,V'}\), to obtain
    the numbers of opponents for which we know that we have a winning
    strategy from \(v\).

  \item We repeat this procedure for every possible combination of sets
    \((v', K_{v'}) \in \W^i\), by varying the witnesses
    ${(K_{v'})}_{v' \in V}$.
    Let us emphasize that we do not take the union over the
    possible \(K_{v'}\). Instead, we construct a set containing all
    the sets yielded by the above procedure.
\end{itemize}

\begin{definition}\label{def:symbolic-alt}
  For $(v, a) \in \vertices \times \alphabet$ and $S \subseteq
  \verticesE$, let
  \(\KPredAlt{v,a}{S} \) be the operator on \(\subsets{\natnonneg}\),
  and let \(\PredAlt{S}\) be the operator on \(\subsets{\verticesE}\)
  defined as
  \begin{gather*}
    \KPredAlt{v,a}{S} = \left\{ \bigcup_{V' \subseteq V}
      \left(\opponents{v,a,V'} \cap \bigcap_{v' \in V'} K_{v'}\right)
      \ \middle\vert\
    \forall v' : (v', K_{v'}) \in S \right\}
    \\
    \PredAlt{S} = \bigcup_{v\in \vertices} \bigcup_{a\in
    \En(v)}\{(v,K) \mid K \in \KPredAlt{v,a}{S}\} \enspace.
  \end{gather*}
  Finally, the sets \({(\WAlt^i)}_{i \in \nat} \subseteq \verticesE\)
  are defined inductively:
  \begin{align*}
    {\displaystyle \WAlt^0} &{\displaystyle = \{(\target, \natnonneg)\} \cup
    \{(v, \emptyset) \mid v \in \vertices \setminus \{\target\}\}}
    \\
    {\displaystyle \forall i \geq 0 : \WAlt^{i+1}} &{\displaystyle =
      \WAlt^i \lub
    \PredAlt{\WAlt^i}.}
  \end{align*}
\end{definition}

For the \(\PredAlt{\cdot}\) operator, one can directly show termination.

\begin{restatable}{lemma}{predAltTermination}\label{lemma:PredAlt:termination}
  The sequence \({(\WAlt^i)}_{i \in \nat}\) stabilizes.
\end{restatable}
\begin{proof}
  Let \(\latticeK\) be the minimal set such that
  \begin{itemize}
    \item \(\natnonneg \in \latticeK\), \(\emptyset \in \latticeK\),
    \item for every \(v, v' \in \vertices\) and \(a \in \alphabet\),
      \(\opponents{v, a, v'} \in \latticeK\), and
    \item for every \(K, K' \in \latticeK\), one has
      \(K \cup K' \in \latticeK\), \(K \cap K' \in \latticeK\), and
      \(K \setminus K' \in \latticeK\),
  \end{itemize}
  and let \(\verticesEK = \{(v, K) \in \verticesE \mid K \in \latticeK\}\).
  One can show that \(\latticeK\) and \(\verticesEK\) are finite
  (using algebraic arguments and the fact that we have
  finitely many \enquote{atoms}).

  Whenever $S$ is a subset of $\verticesEK$, then
  $\KPredAlt{v,a}{S}$ belongs to $\latticeK$ and $\PredAlt{S}$ is included
  in $\verticesEK$. $\verticesEK$ is thus stable by the operator. Since
  $\WAlt^0 \subseteq \verticesEK$, the sets $\WAlt^i$ all belong to
  $\verticesEK$.  As $\latticeK$ is finite, the fixed-point
  computation of the sequence ${(\WAlt^i)}_{i \in \nat}$ will thus
  converge in finitely many steps.
\end{proof}
Note that proving that \(\KPred{v, a}{S}\) belongs to
\(\latticeK\) is not an easy feat, and will be obtained as a
consequence of \(\W\) and \(\WAlt\) computing the same sets. This
justifies the introduction of \(\KPredAlt{\cdot}{\cdot}\).

We can also show a monotonicity property with respect to
the preorder.

\begin{restatable}{lemma}{predAltMonotonicity}\label{lemma:PredAlt}
  For any $(v,a) \in \vertices \times \alphabet$ and two sets
  $S,S' \subseteq \verticesE$ such that $S \parOrdersim S'$,
  \begin{itemize}
    \item
      \(\KPredAlt{v,a}{S} \subseteq \expand{(\KPredAlt{v,a}{S'})}\); and
    \item
      \(\PredAlt{S} \parOrdersim \PredAlt{S'}\).
  \end{itemize}
\end{restatable}

\begin{example}\label{ex:WAlt}
  We illustrate this alternative approach on the game
  \(\game = (\arena, t)\) of
  \Cref{fig:arena:nondeterministic}. Initially
  $\WAlt^0 = \{(s, \emptyset), (v, \emptyset), (x_1, \emptyset), (x_2,
  \emptyset), (y_1, \emptyset), (y_2, \emptyset), (t, \natnonneg)\}$.

  We focus on computing \(\KPredAlt{v, c}{\WAlt^0}\).
  Given \(\WAlt^0\), there is a unique choice
  for \(K_{v'}\) for each \(v' \in \vertices\). Now, since the only
  successors of $v$ by a $c$-edge are $s$ and
  $t$, for every \(V' \subseteq \vertices\),
  if \(\opponents{v, c, V'}\) is not empty, then it must be that
  $V' \subseteq \{s,t\}$.
  When \(s \in V'\), we immediately conclude that
  \(\bigcap_{v' \in V'} K_{v'} = \emptyset\).  It thus remains to
  treat \(V' = \{t\}\), which yields
  \(\opponents{v, c, t} \cap K_t = \natsubs{>2}\).  Hence,
  \(\KPredAlt{v, c}{\WAlt^0} = \{\emptyset \cup \natsubs{>2}\} =
  \{\natsubs{>2}\}\).
  Analogously, one can compute
  \(\KPredAlt{y_1, b}{\WAlt^0} = \KPredAlt{y_2, b}{\WAlt^0} =
  \{\natnonneg\}\). In the end,
  \[
    \WAlt^1 = \{(s, \emptyset), (v, \natsubs{>2}), (x_1, \emptyset),
      (x_2, \emptyset),
    (y_1, \natnonneg), (y_2, \natnonneg), (t, \natnonneg)\}.
  \]
  Using similar steps, we obtain that
  \(\KPredAlt{x_1, b}{\WAlt^1} = \{\natsubs{\neq 2}\}\) and
  \(\KPredAlt{x_2, b}{\WAlt^1} = \{\natsubs{\neq 1}\}\).
  Thus,
  \[
    \WAlt^2 = \{(s, \emptyset), (v, \natsubs{>2}), (x_1, \natsubs{\neq
      2}), (x_2, \natsubs{\neq 1}), (y_1, \natnonneg), (y_2,
    \natnonneg), (t, \natnonneg)\}.
  \]
  We now
  observe that $x_1$ and $x_2$ are the only targets of an edge labeled $a$
  from $v$. Therefore, to
  compute
  \(\KPredAlt{v, a}{\WAlt^2}\), we focus on  \(\{x_1\}, \{x_2\}\),
  and \(\{x_1, x_2\}\).
  Due to the non-determinism,
  \(\opponents{v, a, \{x_1\}} = \opponents{v, a, x_1}
    \setminus \bigcup_{v'
  \neq x_1} \opponents{v, a, v'} = \emptyset\).  Likewise, \(\opponents{v,
  a, \{x_2\}} = \emptyset\).
  Finally, \(\opponents{v, a, \{x_1, x_2\}} = \natnonneg\), but
  \(\bigcap_{v' \in \{x_1,x_2\}} K_{v'} = \{1\} \cap \{2\} = \emptyset\).
  Thus, \(\KPredAlt{v, a}{\WAlt^2} = \{\emptyset\}\) and
  \(\WAlt^3 = \WAlt^2\).

  Observe that, for every index \(i\), \(\WAlt^i = \W^i\) (with
  \(\W^i\) computed in \Cref{ex:WK}). This is not a coincidence, and
  we prove it in the next section.
\end{example}

\subsection{Two fixed-point
computations?}\label{subsec:twofixedpointcomputations}
As we shall see, the two previous approaches contain exactly the
same sets.

\begin{restatable}{theorem}{fixpointequality}\label{thm:same-fixed-point}
  For all \(i \in \nat\), \(\W^i = \WAlt^i\). \\ The sequences
  ${(\W^i)}_{i \in \nat}$ and ${(\WAlt^i)}_{i \in \nat}$ stabilize, and
  \( \W^\infty = \WAlt^\infty\).
\end{restatable}

\begin{proof}[Sketch]
  We give here the main ingredients of the proof, starting with a strong
  relationship between
  $\KPred{\cdot,\cdot}{\cdot}$ and $\KPredAlt{\cdot,\cdot}{\cdot}$:
  for every \(v \in \vertices\), \(a \in \En(v)\) and
  $S \subseteq \verticesE$, we have
  \(\KPred{v, a}{S} \subseteq \KPredAlt{v, a}{S}\)
  and
  \(\KPredAlt{v,a}{S} \subseteq \KPred{v,a}{\expand{S}}\).

  Notice that the second inclusion requires taking the downward
  closure of $S$. Consider for instance an edge
  $v \xrightarrow{a,[2,5]} t$ with the assumption that
  $(t,\natnonneg) \in S$. Then $K_t = \natnonneg$ and
  $\opponents{v,a,t} \cap K_t = [2,5]$, whereas $(t,[2,5]) \notin
  S$. This also justifies the downward closure in the definition of
  ${(\W^i)}_{i \in \nat}$.

  If \(S\) is downward-closed, we derive
  \(\KPred{v, a}{S} = \KPredAlt{v, a}{S}\).
  From there, it should be clear that
  \(\Pred{S} = \PredAlt{S}\).

  Clearly, $\WAlt^0 = \W^0$.  Unfortunately, \(\W^0\)
  is not downward-closed: it contains
  $(t,\natnonneg)$ and none of its dominated elements. However,
  we still manage to prove by induction on $i$ that for every
  $i\in \nat$, $\WAlt^i = \W^i$:
  \begin{align*}
    \WAlt^{i+1} &= \WAlt^i \lub \PredAlt{\WAlt^i}
    = \W^i \lub \PredAlt{\expand{\WAlt^i}}
    =\W^i \lub \Pred{\expand{\W^i}}
    = \W^{i+1}.
  \end{align*}

  By \Cref{lemma:PredAlt:termination}, the sequence
  \({(\WAlt^i)}_{i \in \nat}\) stabilizes, and so does
  \({(\W^i)}_{i \in \nat}\) thanks to the equality.
\end{proof}

Combining Theorems~\ref{thm:symbolic:strategy:region}
and~\ref{thm:same-fixed-point} we obtain our final result:
\begin{theorem}\label{thm:main}
  The antichains $\W^\infty$ and $\WAlt^\infty$ describe
  the maximal elements of the winning region for Eve in $\kGame$.
\end{theorem}

\section{Experimental evaluation}\label{sec:experiments}

We implemented our new algorithms \(\WAlt\) and \(\W\) based on
antichains, and the DFS algorithm of~\cite{BertrandBM19} in an
open-source C++ tool called
\textsf{ParaGraphs}\footnote{\url{https://gitlab.inria.fr/x-staquet/paragraphs}}.
It currently supports constraints given as finite unions of intervals,
but is designed in a modular way and makes it easy to define new
types (for instance semi-linear sets), as well as specific
implementations of the algorithms for these types.

\paragraph{Implementation of \(\W\).}

Recall that \(\KPred{\cdot, \cdot}{\cdot}\) requires to test every subset of
\(\natnonneg\), \ie, infinitely many sets.
However, as introduced in the proof of \Cref{lemma:PredAlt:termination}, any
set produced by \(\KPredAlt{\cdot, \cdot}{\cdot}\) belongs to a finite
subset \(\latticeK\) of \(\subsets{\natnonneg}\), which is a lattice, that
can be built from the constraints of the arena.
Moreover, by \Cref{thm:main}, for any \(i\), \(\W^i = \WAlt^i\).
Thus, it suffices to only consider sets that belong to \(\latticeK\).
That is, we can test finitely many sets.

In \textsf{ParaGraphs}, \(\latticeK\) is constructed from the constraints
appearing on the arena such that we can iterate over the elements, starting from
\(\natnonneg\), in an efficient way. This renders the implementation of
\(\W\) straightforward: to compute \(\KPred{v, a}{S}\), we iterate over
\(\latticeK\) and test whether the current element \(K\) is such that
\((v', K \cap \opponents{v, a, v'}) \notin S\), in which case we consider the
children of \(K\) in the lattice. Otherwise, \(K\) is a maximal element of the
set to return, and it is useless to explore the sublattice rooted at \(K\).
Furthermore, we keep track of which nodes can be ignored (because
they are already covered by a maximal element), in order to
avoid useless computations. Finally, we highlight that there is no
need to compute \(\expand{\W^i}\) to obtain \(\W^{i+1}\).
Indeed, we are only interested in knowing whether
\((v', K \cap \opponents{v, a, v'}) \in S\), which can be performed
directly on the maximal elements of the antichain.

\paragraph{Implementation of \(\WAlt\).}

Given its definition, the implementation of \(\WAlt\) is even more
direct, although heavily
optimized.
For instance, an antichain for \(\WAlt\) is represented as a
dynamically sized array;
to reduce time spent allocating and freeing it, the same array is reused
at each call of \(\KPredAlt{\cdot, \cdot}{\cdot}\).
Importantly, \(\WAlt\) does not require the computation of \(\latticeK\).

\paragraph{Implementation of the DFS algorithm.}

The DFS algorithm in~\cite{BertrandBM19} does not compute the
whole winning region, but determines whether Eve has a winning
strategy in \(\kGame\) from a given initial vertex
\((v_0, \natnonneg)\).  It performs a decomposition of the knowledge
game into sub-games, each of which is defined for a pair \((v, K)\)
and is a restriction of \(\kGame\) to vertices \((v', K, a)\) and
\((v', K')\) that are reachable from \((v, K)\) via vertices with exactly the
same knowledge set \(K\).  That is, except in the \enquote{exit
vertices}, the knowledge set is not refined within a sub-game.
Each sub-game can be solved in polynomial time in the
size of \(\game\). The algorithm tags each node of an exponential-size
tree \(\tree\) of sub-games with \Win or \Lose in a DFS fashion. This
DFS exploration is implemented in
\textsf{ParaGraphs}, together with the resolution of each sub-game by
standard (optimized) attractor computation.

\paragraph{Research statements}

Our experimental evaluation aims to address the following research
questions:
\begin{description}
  \item[Q1]
    Between \(\W\) and \(\WAlt\), which approach is the most efficient?
  \item[Q2]
    Are there arenas on which the antichain-based methods are more
    efficient than the DFS algorithm?
\end{description}

\subsection{Experimental results}

The experiments were carried out on two types of parameterized arenas:
on the one hand, scalable games derived from an illustrative example
in~\cite{BertrandBM19}, and on the other hand games obtained from
quantified Boolean formulas (QBF). Benchmarks were executed on a Dell
Precision 7780, running Debian 13.  \textsf{ParaGraphs} was compiled
with g++ 14.2.0.

\paragraph{Synthetic arenas.} To exercise
our tool \textsf{ParaGraphs}, we considered four families of scalable
instances, with the objective to cover a wide range of
cases. Each family has a parametric size: the parameter
is $n$, and the arena size (precisely the number of vertices and of
different constraints) is linear in $n$; some are deterministic
(D), some non-deterministic (ND), some are winning for Eve (W), some
are losing for Eve (NW). They are given in \Cref{app:experiments}.

\Cref{tab:experiments:solve:games} compares the running time of the three
algorithms, with and
without the time needed to construct the lattice for \(\W\).
To be more specific, the instances with
\(n = 11\) (resp.\ \(n = 12\)) generate a lattice consisting of 4096
(8192) sets that is constructed in about 15s (70s).
Computing larger
lattices exceeds our timeout of five
minutes.  We again highlight that the algorithm from~\cite{BertrandBM19}
solely determines whether Eve has a winning strategy, while our new
symbolic algorithms compute the whole winning region, thus returning a
more complete information.

\begin{table}[t]
  \caption{Time and memory taken by the three algorithms.  Timeout is
    set to five
    minutes. For \(\W\), we list time with and without the construction of the
    lattice, which counts towards the time limit.  The best time and
    memory consumption of each
  row is highlighted in bold.}\label{tab:experiments:solve:games}
  \centering
  \renewcommand{\multirowsetup}{\raggedleft}
  \renewcommand{\arraystretch}{0.8}
  \footnotesize
  \begin{tabular}{@{}
      l@{}r
      @{\hspace{10pt}}r@{\hspace{10pt}}r
      @{\hspace{10pt}}r
      @{\hspace{10pt}}r
      @{\hspace{10pt}}r@{\hspace{10pt}}r@{\hspace{10pt}}r
    @{}}
    \toprule
    \multicolumn{2}{c}{Game} &
    \multicolumn{4}{c}{Time (ms)} &
    \multicolumn{3}{c}{Memory (kB)}
    \\
    \cmidrule(r){1-2}
    \cmidrule(rl){3-6}
    \cmidrule(l){7-9}
    Name & Value of \(n\) &
    \multicolumn{2}{c}{\(\W\)} &
    \multirow{2}*{\(\WAlt\)} &
    \multirow{2}*{DFS} &
    \multirow{2}*{\(\W\)} &
    \multirow{2}*{\(\WAlt\)} &
    \multirow{2}*{DFS}
    \\
    \cmidrule(rl){3-4}
    & & Only & Total
    \\
    \midrule
    \multirow[c]{6}{*}{D-NW-1} & 10 & 850 & 3784 &
    \bfseries 67 & 108704 &
    7005.60 & 5067.20 & \bfseries 4642.40 \\
    & 11 & 3581 & 17312 & \bfseries 246 & timeout &
    10365.60 & \bfseries 5744.80 & timeout \\
    & 12 & 16017 & 82165 & \bfseries 961 & timeout &
    17808.80 & \bfseries 7768.00 & timeout \\
    & 13 & timeout & timeout & \bfseries 3817 & timeout &
    timeout & \bfseries 11936.80 & timeout \\
    & 14 & timeout & timeout & \bfseries 15269 & timeout &
    timeout & \bfseries 20584.00 & timeout \\
    & 15 & timeout & timeout & \bfseries 66657 & timeout &
    timeout & \bfseries 40322.40 & timeout \\
    \midrule
    \multirow[c]{6}{*}{D-W-1} & 10 & 608 & 3581 & 68 &
    \bfseries <1 &
    7491.60 & 5102.40 & \bfseries 4735.60 \\
    & 11 & 2559 & 16696 & 250 & \bfseries <1 &
    11694.40 & 6116.40 & \bfseries 4686.80 \\
    & 12 & 11493 & 80720 & 975 & \bfseries <1 &
    19790.80 & 8058.00 & \bfseries 4708.00 \\
    & 13 & timeout & timeout & 3707 & \bfseries <1 &
    timeout & 12733.60 & \bfseries 4674.00 \\
    & 14 & timeout & timeout & 15105 & \bfseries <1 &
    timeout & 22339.60 & \bfseries 4738.80 \\
    & 15 & timeout & timeout & 66224 & \bfseries <1 &
    timeout & 43618.40 & \bfseries 4696.40 \\
    \midrule
    \multirow[c]{6}{*}{D-NW-2} & 10 & 211 & 3345 & 30 &
    \bfseries 2 &
    6391.20 & 4568.80 & \bfseries 4538.40 \\
    & 11 & 747 & 15574 & 85 & \bfseries 2  &
    8881.60 & 4804.80 & \bfseries 4716.00\\
    & 12 & 2720 & 74282 & 276 & \bfseries 3  &
    13774.00 & 5253.20 & \bfseries 4695.60\\
    & 13 & timeout & timeout & 945 & \bfseries 30 &
    timeout & 6324.00 & \bfseries 4541.60 \\
    & 14 & timeout & timeout & 3571 & \bfseries 41 &
    timeout & 8076.40 & \bfseries 4614.80 \\
    & 15 & timeout & timeout & 12425 & \bfseries 54 &
    timeout & 12889.60 & \bfseries 4616.00 \\
    \midrule
    \multirow[c]{6}{*}{D-W-2} & 10 & 473 & 3569 & 104 &
    \bfseries <1 &
    6517.20 & 4803.60 & \bfseries 4540.00 \\
    & 11 & 1764 & 16680 & 323 & \bfseries <1 &
    9087.60 & 5216.80 & \bfseries 4725.20 \\
    & 12 & 6668 & 78194 & 1065 & \bfseries 2 &
    14324.00 & 5982.40 & \bfseries 4673.20 \\
    & 13 & timeout & timeout & 3903 & \bfseries 13 &
    timeout & 7767.20 & \bfseries 4728.40 \\
    & 14 & timeout & timeout & 14157 & \bfseries 12 &
    timeout & 11820.40 & \bfseries 4718.80 \\
    & 15 & timeout & timeout & 53060 & \bfseries 13 &
    timeout & 19710.80 & \bfseries 4651.60 \\
    \midrule
    \multirow[c]{6}{*}{ND-NW} & 10 & 9 & 2935 &
    \bfseries <1 & 102195 &
    6830.00 & 4666.40 & \bfseries 4400.00 \\
    & 11 & 23 & 14230 & \bfseries 2 & timeout &
    10323.60 & \bfseries 4777.60 & timeout \\
    & 12 & 56 & 67943 & \bfseries 6 & timeout &
    16984.40 & \bfseries 5126.80 & timeout \\
    & 13 & timeout & timeout & \bfseries 9 & timeout &
    timeout & \bfseries 5778.00 & timeout \\
    & 14 & timeout & timeout & \bfseries 20 & timeout &
    timeout & \bfseries 7258.40 & timeout \\
    & 15 & timeout & timeout & \bfseries 40 & timeout &
    timeout & \bfseries 10236.80 & timeout \\
    \bottomrule
  \end{tabular}
\end{table}

Based on these results, we can answer our two research
questions for synthetic arenas. More specifically for {\bf Q1},
\(\WAlt\) is more efficient than \(\W\), as
observed by the number of timeouts. This difference majorly comes from
the construction of \(\latticeK\): computing the lattice in advance
already costs more time than the iterations of \(\WAlt\).

Still for synthetic arenas, the answer to {\bf Q2} is
\enquote{yes}, as, on D-NW-1, \(\WAlt\) greatly outperforms the DFS
algorithm of~\cite{BertrandBM19}.  Interestingly, the latter vastly
outperforms the former when Eve has a winning strategy (D-W-1 and
D-W-2) and for D-NW-2. For the two winning cases, this is simply due
to the fact that the tree exploration is much quicker:
since there is a winning strategy, we explore less sub-games than
compared to D-NW-1 and D-NW-2.  For instance, for \(n = 10\), 31
sub-games are considered for D-W-1, and 19728201 for D-NW-1. Moreover,
the structure of D-NW-2 already induces a small tree: only 135
sub-games are needed (and 34 for D-W-2). Both \(\W\) and \(\WAlt\)
need 22 iterations to compute the whole winning
region.\footnote{Recall that, for every \(i\), \(\WAlt^i = \W^i\), \ie, it is
expected that both algorithms need the same number of iterations.}
Finally, notice that \(\WAlt\) is the only algorithm that managed to
finish all instances. Thus, it seems to scale well to larger and more
complex arenas.

We measured the memory consumption of the algorithms,
which is also reported in \Cref{tab:experiments:solve:games}.
We observe that D-W-1 is
the arena that induces the highest memory consumption for every algorithm.
The fact that the antichain-based methods require more
space than the DFS algorithm is not surprising. Indeed, one may
need to store exponentially many maximal elements inside an
antichain, \ie, our new algorithms may require exponential space in
the worst case.

\begin{figure}[t]
  \centering
  \begin{tikzpicture}[
      /pgfplots/table/x expr={(\thisrow{Variables}==7 &&
      \thisrow{Clauses}>=60)?\thisrow{Clauses}:nan},
    ]
    \begin{axis}[
        footnotesize,
        width = 1\textwidth,
        height = 90pt,
        axis x line = bottom,
        axis y line = left,
        xlabel = {Number of clauses},
        xlabel near ticks,
        ylabel = {Time (ms)},
        ylabel near ticks,
        ylabel shift = {-5pt},
        ymajorgrids = true,
        xtick = data,
        enlarge y limits = 0.1,
        enlarge x limits = 0.1,
        nodes near coords={\pgfmathprintnumber[1000
        sep={\,}]\pgfplotspointmeta},
      ]
      \addplot+ [
        Red,
        only marks,
        mark=*,
      ]
      table [
        y expr = \thisrow{W time total (ms)},
        col sep = comma,
      ]
      {qbf_solving.csv}
      ;
      \addplot+ [
        Blue,
        only marks,
        mark=triangle*,
      ]
      table [
        y expr = \thisrow{WAlt time (ms)},
        col sep = comma,
      ]
      {qbf_solving.csv}
      ;
      \addplot+ [
        Indigo,
        only marks,
        mark=square*,
      ]
      table [
        y expr = \thisrow{Explicit time (ms)},
        col sep = comma,
      ]
      {qbf_solving.csv}
      ;
    \end{axis}
  \end{tikzpicture}
  \caption{Time taken by \(\W\), \(\WAlt\) on randomly generated QBF
    instances with seven variables. For each number of
    clauses, we generate 50 instances, and we compute the average,
    ignoring the runs that timed out. The timeout was set to five
    minutes. On those instances, the DFS
    algorithm from~\cite{BertrandBM19} systematically reached the
    timeout limit, while \(\W\) and \(\WAlt\)
    always succeed. Red circles are for \(\W\) and blue triangles for
    \(\WAlt\). Values are displayed  next to the
  symbol.}\label{fig:experiments:solve:qbf}

  \begin{tikzpicture}[
      /pgfplots/table/x expr={(\thisrow{Variables}==7 &&
      \thisrow{Clauses}>=60)?\thisrow{Clauses}:nan},
    ]
    \begin{axis}[
        footnotesize,
        width = 1\textwidth,
        height = 85pt,
        axis x line = bottom,
        axis y line = left,
        xlabel = {Number of clauses},
        xlabel near ticks,
        ylabel = {Memory (MB)},
        ylabel near ticks,
        ylabel shift = {-5pt},
        ymajorgrids = true,
        xtick = data,
        enlarge y limits = 0.1,
        enlarge x limits = 0.1,
        nodes near coords={\pgfmathprintnumber[1000
        sep={\,}]\pgfplotspointmeta},
      ]
      \addplot+ [
        Red,
        only marks,
        mark=*,
      ]
      table [
        y expr = \thisrow{W memory (kB)} / 1000,
        col sep = comma,
      ]
      {qbf_memory.csv}
      ;
      \addplot+ [
        Blue,
        only marks,
        mark=triangle*,
      ]
      table [
        y expr = \thisrow{WAlt memory (kB)} / 1000,
        col sep = comma,
      ]
      {qbf_memory.csv}
      ;
      \addplot+ [
        Indigo,
        only marks,
        mark=square*,
      ]
      table [
        y expr = \thisrow{Explicit memory (kB)} / 1000,
        col sep = comma,
      ]
      {qbf_memory.csv}
      ;
    \end{axis}
  \end{tikzpicture}
  \caption{Memory taken by \(\W\), \(\WAlt\), and the DFS algorithm on
    randomly generated QBF instances, with seven variables. For each number of
    clauses, we generate 50 instances, and we compute the average,
    ignoring the runs that timed out. All executions of the DFS
    algorithm on those instances timed
  out.}\label{fig:experiments:solve:qbf:memory}
\end{figure}

\paragraph{Arenas obtained from quantified Boolean formulas.} Our
tool \textsf{ParaGraphs} also handles quantified Boolean formulas
(QBF) and checks their satisfiability by reduction to solving a
concurrent parameterized reachability game, as in the \PSPACE-hardness
proof~\cite[Prop.\ 13]{BertrandBM19}.

Interestingly, given the constraints in arenas generated from QBF
instances, the constructing the lattice \(\latticeK\) (whose size only
depends on the number of variables) can be made very efficient: a
clever order on insertions  always avoids costly equality checks
(with elements already in the lattice). For formulas with 6 (resp.\ 7)
variables, the lattice, whose size is 8192 (resp.\ 32768), is obtained
in about 575ms (resp.\ 9s). Observe that our tailored construction
permits us to build larger lattices. We exercised the implementation
of the three algorithms on bigger arenas than the previous ones by
randomly generating 50 formulas with \(n\) variables and \(m\)
clauses, for every \(n \in \{4,5,6,7\}\) and
\(m \in \{50,60,\dotsc,100\}\). We focus here on the instances with 7
variables and defer to the appendix the other instances.

The answers to {\bf Q1} and {\bf Q2} on benchmarks
generated from QBF formulas are similar to the ones on synthetic
arenas.  First, both \(\W\) and \(\WAlt\) manage to finish all
instances, while the DFS algorithm systematically reached
the time limit.  Second, \Cref{fig:experiments:solve:qbf} gives the
time taken by the \(\W\) and \(\WAlt\) algorithms on the formulas with
7 variables and with 60 to 100 clauses. For clarity, the exact values
are displayed next to each point. We observe that \(\WAlt\) is the
most efficient algorithm, by a margin of three orders of
magnitude. The memory consumption is given in
\Cref{fig:experiments:solve:qbf:memory}. Similarly to the synthetic
arenas, \(\WAlt\) needs less memory than \(\W\), and the algorithm
from~\cite{BertrandBM19} it the most efficient memory-wise. However,
we do not measure the memory consumption of the instances that timed
out.

\paragraph{Conclusion.}
We can give the following answers to our research
questions:
\begin{description}
  \item [A1]
    \(\WAlt\) is more efficient than \(\W\), both time- and memory-wise.
  \item [A2]
    There are arenas on which \(\WAlt\) outspeeds the DFS algorithm.
    The memory consumption of the latter is always lower than those
    of the former.
    However, our experiments are not sufficient to deduce syntactic
    characteristics
    of arenas on which \(\WAlt\) outperforms the DFS algorithm.
    Moreover, a potential optimization that would allow a more direct
    comparison between the DFS algorithm and \(\WAlt\) would be to
    stop computing the winning region as soon as \((v_0,
    \natnonneg)\) is encountered.
\end{description}

\section{Conclusion}\label{sec:conclusion}

This paper proposes a symbolic approach to compute the winning region
in parameterized concurrent reachability games. It is based on
antichains to represent sets of knowledges in an effective way. After
providing a fixed-point characterization of the winning region, we
gave two effective algorithms to compute the least fixed-point, one
that iterates over the elements of a lattice (which is finite when the
arena is finite), and one that constructs the knowledge sets on the
fly.  We implemented both methods, as well as the DFS algorithm
from~\cite{BertrandBM19}, in an open-source tool called
\textsf{ParaGraphs}, and performed benchmarks on multiple examples,
including arenas obtained from quantified Boolean formulas.

\bibliographystyle{eptcs} \bibliography{abbreviations,references}

\appendix
\newpage
\section{Details on Section~\ref{sec:symbolic}}
\subsection{Missing proofs of
Subsection~\ref{subsec:antichains}}\label{app:proofs:antichain}

\propExpandReduce*

\begin{proof}
  Let \(L \subseteq X\). We prove the items one by one.
  \begin{enumerate}
    \item
      The fact that \(\reduce{L}\) is an antichain
      for \((X, \parOrder)\) is straightforward from the
      definitions.

    \item
      We show that \(\Dom(L) = \Dom(\expand{L})\). We start with
      \(\Dom(L) \subseteq \Dom(\expand{L})\). Let \(x \in \Dom(L)\).
      That is, \(x \in X\) and there exists some \(y \in L\) such that
      \(x \parOrder y\) and \(x \neq y\). Since \(y \in L\), we necessarily
      have that \(y \in \expand{L} = L \cup \Dom(L)\). Hence,
      \(x \in \Dom(\expand{L})\).

      We now turn to \(\Dom(\expand{L}) \subseteq \Dom(L)\).
      Let \(x \in \Dom(\expand{L})\), \ie, \(x \in X\) and there exists some
      \(y \in \expand{L}\) such that \(x \parOrder y\) and \(x \neq y\).
      As \(y \in \expand{L} = L \cup \Dom(L)\), we have two cases:
      \begin{itemize}
        \item
          If \(y \in L\), then we immediately get \(x \in \Dom(L)\).
        \item
          If \(y \notin L\), it still follows that \(y \in \Dom(L)\), meaning
          that there must exist some
          \(z \in L\) such that \(y \parOrder z\) and \(y \neq z\).
          By transitivity, \(x \parOrder z\). It remains to show that
          \(x \neq z\). Towards a contradiction, assume \(x = z\).
          Then, we have \(x \parOrder y \parOrder z \parOrder x\)
          (by reflexivity of \(\parOrder\)), leading to \(x = y\)
          (by antisymmetry of \(\parOrder\)). We thus obtain a contradiction
          with \(x \neq y\).
      \end{itemize}

    \item
      We prove that \(\expand{L} \cup \expand{L'} = \expand{(L \cup L')}\).
      Let \(x \in \expand{L} \cup \expand{L'}\). We focus here on
      \(x \in \expand{L}\); the other case uses similar arguments.
      Since \(x \in \expand{L}\), either \(x \in L\), meaning that
      \(x \in L \cup L'\), \ie, \(x \in \expand{(L \cup L')}\); or
      \(x \in \Dom(L)\), meaning that \(x \in \Dom(L \cup L')\)
      (that can easily be seen by definition of \(\Dom\)), and, so,
      \(x \in \expand{(L \cup L')}\).

      Now, let \(x \in \expand{(L \cup L')}\). Either \(x \in L \cup L'\),
      from which it is easy to derive \(x \in \expand{L} \cup \expand{L'}\);
      or \(x \in \Dom(L \cup L')\), in which case there exists
      \(y \in L \cup L'\) such that \(x \parOrder y \land x \neq y\).
      If \(y \in L\), then \(x \parOrder y\) implies that \(x \in \expand{L}\),
      and, so, \(x \in \expand{L} \cup \expand{L'}\). We do likewise when
      \(y \in L'\).

    \item
      We now prove that \(\reduce{L} = \reduce{\expand{L}}\) via the
      following equalities.
      \begin{align*}
        \reduce{\expand{L}}
        &= \expand{L} \setminus \Dom(\expand{L})
        & \text{by definition}
        \\
        &= \expand{L} \setminus \Dom(L)
        & \text{by the previous item}
        \\
        &= (L \cup \Dom(L)) \setminus \Dom(L)
        & \text{by definition}
        \\
        &= L \setminus \Dom(L) \cup \Dom(L) \setminus \Dom(L)
        = L \setminus \Dom(L)
        = \reduce{L}.
      \end{align*}

    \item Let us show that \(\expand{L} = \expand{\reduce{L}} \iff
      \forall x \in L, \exists y \in \reduce{L} : x \parOrder y\).
      Observe that since \(\reduce{L} \subseteq L\), the inclusion
      \(\expand{\reduce{L}} \subseteq \expand{L}\) always
      holds. Let us prove both implications in turn.
      \begin{itemize}
        \item Assume that \(\expand{L} = \expand{\reduce{L}}\), and let
          \(x\in L\). By definition of the expanded form,
          \(x\in \expand{L}\) and by assumption,
          \(x\in \expand{\reduce{L}}\). Again using the definition of
          the expanded form, we deduce that either \(x\in \reduce{L}\)
          or \(x \in \Dom(\reduce{L})\). In the first case, one can let
          \(y =x\) to obtain the conclusion. In the second case, by
          definition of the dominated set, there must exist
          \(y \in \reduce{L}\) with \(x \neq y\) such that
          \(x \parOrder y\), which proves the conclusion.
        \item Assume now that
          \(\forall x\in L, \exists y \in \reduce{L}\) such that
          \(x \parOrder y\), and let us prove that
          \(\expand{L} \subseteq \expand{\reduce{L}}\). Let
          \(z \in \expand{L}\). There must exist \(x \in L\) such that
          \(z \parOrder x\). By assumption, there exists
          \(y \in \reduce{L}\) with \(x \parOrder y\). By transitivity
          of the partial order \(\parOrder\), we derive
          \(z \parOrder y\) and thus \(z \in \expand{\reduce{L}}\).\qedhere
      \end{itemize}

    \item
      Finally, we argue that, if \(\expand{L} = \expand{\reduce{L}}\)
      and \(\expand{L'} = \expand{\reduce{L'}}\), then
      \(\expand{\reduce{L \cup L'}} = \expand{(\reduce{L} \cup \reduce{L'})}\).
      Let \(x \in \expand{\reduce{L \cup L'}}\). Either \(x \in
      \reduce{L \cup L'}\), allowing us to derive the following implications:
      \begin{multline*}
        x \in \reduce{L \cup L'}
        \implies
        x \in L \cup L'
        \implies
        x \in \expand{(L \cup L')}
        \iff
        x \in \expand{L} \cup \expand{L'}
        \iff
        x \in \expand{\reduce{L}} \cup \expand{\reduce{L}}
        \\
        \iff
        x \in \expand{(\reduce{L} \cup \reduce{L'})};
      \end{multline*}
      or \(x \in \Dom(\reduce{L \cup L'})\), in which case there exists
      \(y \in \reduce{L \cup L'}\) such that \(x \parOrder y\) and \(x \neq y\).
      Using similar arguments as before, we conclude that \(y \in
      \expand{(\reduce{L} \cup \reduce{L'})}\), meaning that
      \(x \in \expand{(\reduce{L} \cup \reduce{L'})}\), as \(x \parOrder y\).

      Now, let \(x \in \expand{(\reduce{L} \cup \reduce{L'})}
      = \expand{\reduce{L}} \cup \expand{\reduce{L'}}\). We focus on
      \(x \in \expand{\reduce{L}}\), the other case follows the same arguments.
      As \(x \in \expand{\reduce{L}}\), it easily follows that
      \(x \in \expand{\reduce{L \cup L'}}\).
  \end{enumerate}
\end{proof}

\lemmareduce*

\begin{proof}
  Clearly, the equivalences hold whenever \(L\) and
  \(L'\) are both empty. If \(L\) is empty and \(L'\) is not, then
  \begin{enumerate}
    \item
      \(\reduce{L} = \reduce{\emptyset} = \emptyset \neq \reduce{L'}\),
    \item
      \(L' \not\parOrdersim \emptyset\), and
    \item
      \(\expand{L'} \nsubseteq \emptyset\) (as \(L' \neq \emptyset\)).
  \end{enumerate}
  Thus, the equivalences also hold in that case and in the symmetric
  case \(L \neq \emptyset\) and \(L' = \emptyset\). In the rest of
  the proof, we thus assume that both \(L\) and \(L'\) are not
  empty.

  First assume (1), \ie, \(\reduce{L} = \reduce{L'}\). By symmetry,
  it is sufficient to prove $L \parOrdersim L'$. Pick $a \in L$.
  By \Cref{prop:expand_is_expand_reduce}, there
  is $a' \in \reduce{L}$ such that $a \parOrder a'$. By assumption,
  $a' \in \reduce{L'}$, so $a' \in L'$. This proves (2).

  Then assume (2). By symmetry, it is sufficient to prove
  \(L \subseteq \expand{L'}\). Pick $a \in L$. By assumption, there
  exists $a' \in L'$ such that $a \parOrder a'$. It implies
  $a \in \expand{L'}$, which shows (3).

  Then assume (3), that is \(L \subseteq \expand{L'}\) and
  \(L' \subseteq \expand{L}\). By symmetry, it is sufficient to prove
  \(\reduce{L} \subseteq \reduce{L'}\), so pick $a \in
  \reduce{L}$. Since \(\reduce{L} \subseteq L \subseteq \expand{L'}\),
  $a \in \expand{L'}$, hence there is $b \in \reduce{L'}$ such that
  $a \parOrder b$. Since \(\reduce{L'} \subseteq L' \subseteq \expand{L}\),
  $b \in \expand{L}$, hence there is $c \in \reduce{L}$ such that
  $b \parOrder c$. So we have $a \parOrder b \parOrder c$ and
  $a,c \in \reduce{L}$, hence we conclude that $a=c$, therefore $a=b$
  and $a \in \reduce{L'}$. This shows (1).
\end{proof}

\subsection{Missing proofs of Subsection~\ref{subsec:symbolic-algo}}

\fixpointcorrectness*

\begin{proof}
  Let \((v, K) \in \verticesE\) such that \(K \neq \emptyset\).
  We prove both implications separately.

  \proofsubparagraph{\(\implies\).}
  We start by showing that if Eve has a winning strategy \(\kStrat\)
  from \((v, K)\),
  then \((v, K) \in \expand{\W^\infty}\).
  Since \(\kGame\) is a standard two-player turn-based reachability game, we can
  assume that \(\kStrat\) is memoryless.\footnote{The strategy
    \(\kStrat\) is memoryless when, for any \(hv, h'v \in \histories\),
  \(\kStrat(hv) = \kStrat(h'v)\).}

  Let us consider a play \(\rho\) induced by playing \(\kStrat\) from
  \((v, K)\).
  That is, \(\rho\) is of the shape
  \((v_1, K_1) \cdot \kStrat((v_1, K_1)) \cdot
  (v_2, K_2) \cdot \kStrat((v_2, K_2)) \dotsb\)
  such that \((v, K) = (v_1, K_1)\), and for each
  \(i\), \((v_i, K_i) \in \verticesE\).
  Recall that every \(\kStrat(v_i, K_i)\) is an Adam's vertex.
  Since \(\kStrat\) is a winning strategy, there necessarily exists an index
  \(n\) such that \(v_n = \target\), \ie, \((v_n, K_n)\) is a target vertex in
  \(\kGame\).
  We take the history obtained by stopping the play at \((v_n, K_n)\) and
  project it on Eve's vertices (\ie, we drop Adam's vertices in the history).
  We do the same for every plays compatible with \(\kStrat\) and obtain a
  tree \(\tree\).

  Observe that any branch of \(\tree\) is necessarily finite, by construction.
  Let \((v', K')\) be a vertex appearing in \(\tree\).
  Since \(\kStrat\) is a memoryless strategy, \(\kStrat((v', K'))\) always
  picks the same Adam's vertex, yielding the same sub-histories.
  Hence, all subtrees rooted at \((v', K')\) are the same.
  Furthermore, \((v', K')\) cannot appear twice in a branch, as, otherwise,
  Adam has a way to force Eve in the loop, yielding an infinite branch.

  We show that if \((v', K')\) appears at height \(i\) in the tree, then
  \((v', K') \in \expand{\W^i}\).
  Since \((v, K)\) is the root of the tree whose height is, say, \(m\), it
  follows that \((v, K) \in \expand{\W^m}\).
  Hence, \((v, K) \in \expand{\W^\infty}\).
  We prove our claim by induction on \(i\), the height of \((v', K')\).

  \begin{description}
    \item[Base case:] \(i = 0\).
      Since the height of \((v', K')\) is zero, \((v', K')\) is a leaf.
      By construction of \(\tree\), \(v' = \target\), \ie, \((v',
      K')\) is a target
      vertex for Eve in \(\kGame\).
      Hence, \((v', K') \in \expand{\W^0}\), by definition.

    \item[Induction step.]
      Assume the claim holds for \(i \geq 0\), \ie, for every vertex \((u, L)\)
      at height \(i\), we have \((u, L) \in \expand{\W^i}\).
      Suppose that the height of \((v', K')\) is \(i+1\).
      We show that \((v', K') \in \expand{\W^{i+1}}\).
      Let \(a \in \alphabet\) such that \(\kStrat((v', K')) = (v',
      K', a)\), which
      is a vertex of Adam.
      From there, he can go to some vertex \((u, L)\).
      By construction of \(\tree\), \((u, L)\) is at height \(i\) (as the height
      of \((v', K')\) is \(i + 1\)).
      Hence, by the induction hypothesis, \((u, L) \in \expand{\W^i}\).
      Moreover, by definition of \(\kGame\), \(L = K' \cap
      \opponents{v', a, u}\)
      and \(L\) is not empty (as, otherwise, Adam could not have
      played that edge).

      Therefore, for any vertex \((u, K' \cap \opponents{v', a, u})\) such that
      \(K' \cap \opponents{v', a, u} \neq \emptyset\), we have
      \((u, K' \cap \opponents{v', a, u}) \in \expand{\W^{i}}\).
      By definition, it follows that \((v', K') \in \expand{\W^{i+1}}\).
  \end{description}

  \proofsubparagraph{\(\impliedby\).}
  Assume that \((v, K) \in \expand{\W^\infty}\).
  We construct a winning strategy \(\kStrat\).
  First, for every \((v', K') \in \verticesE \setminus \expand{\W^\infty}\),
  select an arbitrary action \(a\) enabled at \((v', K')\) and define
  \(\kStrat((v', K')) = (v', K', a)\).
  Second, for every \((v', K') \in \expand{\W^\infty}\), let \(i\) be
  the smallest
  index such that \((v', K') \in \expand{\W^i}\), \ie, the iteration \(i\) is
  the first time \((v', K')\) was represented by a set \(\W\).
  Then, by definition, there must exist an action \(a\) that is
  enabled at \(v'\)
  and such that for any \(u \in \vertices\) such that
  \(K' \cap \opponents{v', a, u} \neq \emptyset\), it holds that
  \((u, K' \cap \opponents{v', a, u}) \in \expand{\W^{i-1}}\).
  That is, if we play \(a\), we stay in the \enquote{good} part of \(\kGame\)
  and, in a way, we make progress towards a target vertex.
  Hence, let \(\kStrat((v', K')) = (v', K', a)\).

  It remains to show that for any play
  \(\rho = (v_1, K_1) \cdot \kStrat(v_1, K_1) \cdot
  (v_2, K_2) \cdot \kStrat(v_2, K_2) \dotsb\) induced by \(\kStrat\) from
  \((v, K)\), there exists an index \(n\) such that \(v_n = \target\) (\ie,
  \((v_n, K_n)\) is a target vertex in \(\kGame\)).
  Since \((v_1, K_1) = (v, K) \in \expand{\W^\infty}\), there exists some
  \(i_1\) such that \((v_1, K_1) \in \expand{\W^{i_1}}\).
  Furthermore, \(K_2 \neq \emptyset\) by definition of \(\kGame\), and
  \((v_2, K_2) \in \expand{\W^{i_1-1}}\) by construction of
  \(\kStrat\), which means that there exists some \(i_2 \leq i_1 - 1 < i_1\)
  such that \((v_2, K_2) \in \expand{\W^{i_2}}\),
  \((v_3, K_3) \in \expand{\W^{i_2 - 1}}\), and \(K_3 \neq \emptyset\).
  By repeating this procedure, we eventually reach a moment where
  \((v_n, K_n) \in \expand{\W^0}\) and \(K_n \neq \emptyset\).
  By definition of \(\W^0\), it must necessarily hold that \(v_n = \target\),
  as \(K_n \neq \emptyset\).
  Thus, for whatever play \(\rho\) induced by \(\kStrat\), we eventually reach
  a target vertex, implying that \(\kStrat\) is a winning strategy.
\end{proof}

\fixpointcorollary*

\begin{proof}
  Let $i \in \natnonneg$.  Observe that \((v_0, \natnonneg) \in \W^i\)
  if and only if \((v_0, \natnonneg) \in \expand{\W^i}\), as
  \((v_0, \natnonneg)\) is the maximal element with regard to
  \(\parOrder\).  By \Cref{thm:symbolic:strategy:region}, we thus
  obtain that \((v_0, \natnonneg) \in \W^i\) if and only if Eve has a
  winning strategy from \((v_0, \natnonneg)\) in \(\kGame\).  By
  \Cref{thm:knowledge:strategy}, it follows that Eve has a winning
  strategy from \(v_0\) in \(\game\) if and only if there exists
  $i \in \natnonneg$ such that \((v_0, \natnonneg) \in \W^i\).
\end{proof}

\subsection{Missing proof of Subsection~\ref{subsec:alternative}}

\predAltMonotonicity*
\begin{proof}
  Assume $S \parOrdersim S'$, and pick $K \in
  \KPredAlt{v,a}{S}$.
  For each \(v' \in V\), there exists $(v',K_{v'}) \in S$
  such that
  \(K = \bigcup_{V' \subseteq V} \opponents{v,a,V'} \cap
  \bigcap_{v' \in V'} K_{v'}\).
  For every $v' \in V'$, let $K'_{v'}$ be such that $(v',K'_{v'}) \in S'$ and
  $K_{v'} \subseteq K'_{v'}$, and define
  \(K' = \bigcup_{V' \subseteq V} \opponents{v,a,V'} \cap
  \bigcap_{v' \in V'} K'_{v'}\).
  Then, $K' \in \KPredAlt{v,a}{S'}$ and $K \subseteq K'$. Therefore,
  $\KPredAlt{v,a}{S} \subseteq \expand{(\KPredAlt{v,a}{S'})}$.
  As an immediate consequence $\PredAlt{S} \parOrdersim \PredAlt{S'}$.
\end{proof}

As a consequence, we get the following corollary, relating the maximal
elements computed by $\KPredAlt{\cdot,\cdot}{\cdot}$ on a set or its
downward-closure.
\begin{corollary}\label{coro:PredAlt}
  For any $(v,a ) \in \vertices \times \alphabet$ and $S \subseteq \verticesE$,
  \[\reduce{\KPredAlt{v,a}{S}} = \reduce{\KPredAlt{v,a}{\expand{S}}}.\]
\end{corollary}
\begin{proof}
  For any set \(S\), we have
  $S \parOrdersim \expand{S}$ and $\expand{S} \parOrdersim S$.
  By \Cref{lemma:PredAlt}, we get
  \begin{gather*}
    \KPredAlt{v,a}{S} \subseteq \expand{(\KPredAlt{v,a}{\expand{S}})}
    \shortintertext{and}
    \KPredAlt{v,a}{\expand{S}} \subseteq \expand{(\KPredAlt{v,a}{S})}.
  \end{gather*}
  By \Cref{lemma:reduce}, this implies the expected result.
\end{proof}

The latter result justifies that when defining $\WAlt^{i+1}$, there is
no need to apply the downward-closure to $\WAlt^i$, in contrast to
the sequence ${(\W^i)}_{i \in \nat}$.

\subsection{Missing proofs of
Subsection~\ref{subsec:twofixedpointcomputations}}

\fixpointequality*

\begin{proof}
  The proof of this theorem relies on several intermediary results.

  \begin{lemma}\label{lemma:symbolic:know:alt}
    For every vertex \(v \in \vertices\), action \(a \in \En(v)\) and
    set $S \subseteq \verticesE$,
    \[
      \left\{
        \begin{array}{l}
          \KPred{v, a}{S} \subseteq \KPredAlt{v, a}{S},\quad \text{and} \\
          \KPredAlt{v,a}{S} \subseteq \KPred{v,a}{\expand{S}}.
        \end{array}\right.
      \]
    \end{lemma}
    \begin{proof}[Proof of Lemma~\ref{lemma:symbolic:know:alt}]
      Let \(v \in \vertices\), \(a \in \En(v)\) and
      $S \subseteq \verticesE$. We prove the two inclusions.

      Let
      $K \in \KPred{v,a}{S}$. For every $j \in \interval{1}{n}$, let
      $K_j = K \cap \opponents{v,a,v_j}$; in particular,
      $(v_j,K_j) \in S$.
      The following set then belongs to $\KPredAlt{v,a}{S}$:
      \[
        \bigcup_{I \subseteq \interval{1}{n}} \opponents{v,a,I} \cap
        \bigcap_{j \in I} K_j = \bigcup_{I \subseteq \interval{1}{n}}
        \opponents{v,a,I} \cap K = K \quad\text{(by~\Cref{K:complete}
        page~\pageref{K:complete})}.
      \]
      This shows that \(\KPred{v, a}{S} \subseteq \KPredAlt{v, a}{S}\).

      Now, let
      $K \in \KPredAlt{v,a}{S}$. For every $j \in \interval{1}{n}$, let
      $K_j \subseteq \natnonneg$ be such that $(v_j,K_j) \in S$ and
      \(K = \bigcup_{I \subseteq \interval{1}{n}} \opponents{v,a,I} \cap
      \bigcap_{j \in I} K_j\).
      Pick $j \in \interval{1}{n}$, and check that:
      \[
        K \cap \opponents{v,a,v_j} = \bigcup_{I \subseteq \interval{1}{n},
        j \in I} K_j \cap \opponents{v,a,I} \subseteq K_j.
      \]
      We therefore get that
      $(v_j, K \cap \opponents{v,a,v_j}) \in \expand{S}$, from which we conclude
      that $K \in \KPred{v,a}{\expand{S}}$.
    \end{proof}
    Note that the second inclusion requires the downward closure of
    $S$. For instance $v_1 \xrightarrow{a,[2,5]} v_2$ and
    $(v_2,\natnonneg) \in S$. Then $K_2 = \natnonneg$ and
    $\opponents{v_1,a,v_2} \cap K_2 = [2,5]$ is not in $S$. This is also
    the reason why the fix-point defining ${(\W^i)}_{i \in \nat}$ requires
    the downward closure in its definition.

    As a consequence of the previous lemma, we get:

    \begin{corollary}\label{corollary:symbolic:know:alt}
      For every vertex \(v\), action \(a \in \En(v)\) and set
      $S \subseteq \verticesE$ which is downward-closed (that is, $S
      = \expand{S}$),
      \[\KPred{v, a}{S} = \KPredAlt{v, a}{S}\quad \text{and}\quad\Pred{S}
      = \PredAlt{S}.\]
    \end{corollary}

    We now prove the first part of Theorem~\ref{thm:same-fixed-point},
    which states the equality of the sequences.
    \begin{lemma}\label{coro:WAlt=W}
      For every $i \in \nat$, \(\WAlt^i = \W^i\).
    \end{lemma}
    \begin{proof}[Proof of Lemma~\ref{coro:WAlt=W}]
      We show the result by induction on $i$. The equality obviously holds
      for $i=0$.
      Assume now it holds for $i$.
      That is, assume \(\W^i = \WAlt^i\).
      Since all computations of \(\WAlt\) are done in
      the finite lattice \(\latticeK\), any set produced by
      \(\PredAlt{\WAlt^i}\) is finite. Hence,
      \(\expand{\PredAlt{\WAlt^i}} =
      \expand{\reduce{\PredAlt{\WAlt^i}}}\). We have already argued
      that \(\expand{\W^i} = \expand{\reduce{\W^i}}\) and
      \(\expand{\Pred{\expand{\W^i}}} =
      \expand{\reduce{\Pred{\expand{\W^i}}}}\) in
      \Cref{lemma:W:welldefined} and \Cref{cor:W:WNonDecreasing}.
      Hence, we can apply \Cref{prop:lub:reduceLub} in the following:
      \begin{align*}
        \WAlt^{i+1}
        &= \WAlt^i \lub \PredAlt{\WAlt^i}
        & \text{by definition}
        \\
        &= \W^i \lub \PredAlt{\WAlt^i}
        & \text{by induction hypothesis}
        \\
        &= \reduce{\W^i} \lub \reduce{\PredAlt{\WAlt^i}}
        & \text{by \Cref{prop:lub:reduceLub}}
        \\
        &= \reduce{\W^i} \lub \reduce{\PredAlt{\expand{\WAlt^i}}}
        & \text{by \Cref{coro:PredAlt}}
        \\
        &= \reduce{\W^i} \lub \reduce{\Pred{\expand{\W^i}}}
        & \text{by \Cref{corollary:symbolic:know:alt}}
        \\
        &= \W^i \lub \Pred{\expand{\W^i}}
        & \text{by \Cref{prop:lub:reduceLub}}
        \\
        &= \W^{i+1}
        & \text{by definition}
      \end{align*}

      This proves the induction step, and therefore the lemma.
    \end{proof}

    To conclude the proof of \Cref{thm:same-fixed-point}, it remains to
    show that the sequences converge in finite time.
    By \Cref{lemma:PredAlt:termination}, we get that the sequence
    \({(\WAlt^i)}_{i \in \natnonneg}\) stabilizes, and so is the case for
    \({(\W^i)}_{i \in \natnonneg}\) by the shown equality.
  \end{proof}

  \section{Details on the experiments of
  Section~\ref{sec:experiments}}\label{app:experiments}

  The implementation of our tool \textsf{ParaGraphs} heavily relies on
  C++ templates, in order to provide a generic implementation that does
  not need to consider the semantics of the constraints.  It is however
  possible to specialize those templates in order to define a specific
  algorithm for a given type.  For instance, one could implement the
  \PTIME algorithm from~\cite{BertrandBM19} when constraints are known
  to be intervals, or generalize the current
  implementation to semi-linear constraints.

  We now give more information about the construction of \(\latticeK\).
  The lattice \(\latticeK\) is stored in an array, which permits to denote
  an element of the lattice by its index. Its construction is as
  follows.  First, insert the elements \(\natnonneg, \emptyset\), and
  \(\opponents{v, a, v'}\) for every \(v, v' \in \vertices\) and
  \(a \in \alphabet\).  While inserting these elements, we also maintain
  a set of (direct) parents and a set of (direct) children of each
  element.  Initially, the only child of \(\natnonneg\) is \(\emptyset\)
  and the only parent of \(\emptyset\) is \(\natnonneg\).  Then, when we
  insert a set \(K\), we can find the children of \(K\) as follows:
  starting from \(\natnonneg\), we move down the lattice as long as the
  current element includes \(K\).  Since a node may have multiple
  children, multiple branches have to be explored, but we can guarantee
  to not visit the same node more than once (by storing the indices of
  the nodes we already traversed).  When we know the sets of parents and
  children of \(K\), updating the set of children of another element
  \(K'\) is easy: remove every child of \(K'\) that is also a child of
  \(K\) and insert \(K\) as a new child of \(K'\).  (Parents can be
  updated similarly.)  Once these initial elements are added, it
  suffices to iterate over \(K, K' \in \latticeK\) and compute
  \(K \cup K'\), \(K \setminus K'\), and \(K' \setminus K\), and insert
  these new sets if they are not yet present, while updating the parents
  and children as previously.

  \paragraph{Tailored construction for QBF.}

  Finally, we give more details about the construction of \(\latticeK\)
  for QBF instances with \(n\) variables (the clauses have no impact on
  \(\latticeK\)).  We initially add the sets \(\{1\}, \natsubs{>1},
  \dotsc, \{n\}\), and \(\natsubs{>n}\)
  in the lattice.  Then, for each \(i \in \{1, \dotsc, n\}\), we select
  \(\{i\}, \natsubs{>i}\), or \(\emptyset\), and compute the union of these
  selections.  Observe that, if we select \(\natsubs{>i}\) for some \(i\), the
  choice for \(j > i\) becomes irrelevant, as \(\{j\}\) and
  \(\natsubs{>j}\) are both
  included in \(\natsubs{>i}\).  That is, when we pick
  \(\natsubs{>i}\), we do not need to
  make further choices.  This has a nice consequence: every set we
  obtain is unique and is not yet a member of the lattice.  We can thus
  immediately insert the set, which saves us from checking whether it is
  already present.  We maintain the sets of parents and children as in
  the more general construction.  This QBF-tailored algorithm is
  implemented in \textsf{ParaGraphs}.

  \subsection{Details on experimental results}

  \paragraph{Synthetic arenas.}

  \Cref{app:fig:arenas:experiments} gives the synthetic parameterized
  arenas used
  in our benchmarks, while the time needed to construct \(\latticeK\)
  is given in
  \Cref{app:tab:experiments:lattice:games}.
  Finally, \Cref{app:tab:experiments:more} provide the number of iterations and
  generated sets by the \(\W\) and \(\WAlt\) methods, and the number
  of sub-games
  explored by the algorithm from~\cite{BertrandBM19}. Notice that \(\W\) and
  \(\WAlt\) always take the same number of iterations, which is expected by
  \Cref{thm:same-fixed-point}.

  \begin{figure}
    \centering
    \begin{subfigure}{.45\textwidth}
      \centering
      \begin{tikzpicture}[
          game,
          node distance = 25pt and 40pt,
        ]
        \node [state]                               (v)   {\(v\)};
        \node [state, above right=of v]             (x1)  {\(x_1\)};
        \node [state, below right=of v]             (x2)  {\(x_2\)};
        \node [state, right=of x1]                  (y1)  {\(y_1\)};
        \node [state, right=of x2]                  (y2)  {\(y_2\)};
        \node [state, accepting, below right=of y1] (t)   {\(t\)};
        \node [below=of x2]                         (xi)  {\vdots};
        \node [below=of y2]                         (yi)  {\vdots};
        \node [state, below=of xi]                  (xn)  {\(x_n\)};
        \node [state, below=of yi]                  (yn)  {\(y_n\)};

        \path
        (v)   edge                  node [sloped]           {\(a_1\)}       (x1)
        edge                  node [sloped, pos=0.6]  {\(a_2\)}       (x2)
        (x1)  edge                  node                    {\(b,
        \{1\}\)}     (y1)
        (x2)  edge                  node [']                {\(b,
        \{2\}\)}     (y2)
        (y1)  edge                  node [sloped]           {\(b\)}         (t)
        (y2)  edge                  node [', sloped]        {\(b\)}         (t)
        (xn)  edge                  node                    {\(b,
        \{n\}\)}     (yn)
        ;

        \draw [rounded corners = 10pt]
        let
        \p{s} = (x1.west),
        \p{t} = (v.north),
        in
        (\p{s}) -- (\x{t}, \y{s})
        node [above, pos=0.4] {\(b, \natsubs{\neq 1}\)}
        -- (\p{t})
        ;
        \draw [rounded corners = 10pt]
        let
        \p{s} = (x2.west),
        \p{t} = (v.-50),
        in
        (\p{s}) -- (\x{t}, \y{s})
        node [below, pos=0.4] {\(b, \natsubs{\neq 2}\)}
        -- (\p{t})
        ;
        \draw [rounded corners = 10pt]
        let
        \p{s} = (v.-90),
        \p{t} = (xn.160),
        in
        (\p{s}) -- (\x{s}, \y{t})
        -- (\p{t})
        node [above, pos=0.4] {\(a_n\)}
        ;
        \draw [rounded corners = 10pt]
        let
        \p{s} = (xn.-160),
        \p{t} = (v.-130),
        in
        (\p{s}) -- (\x{t}, \y{s})
        node [below, pos=0.4] {\(b, \natsubs{\neq n}\)}
        -- (\p{t})
        ;
        \draw [rounded corners = 10pt]
        let
        \p{s} = (yn.east),
        \p{t} = (t.south),
        in
        (\p{s}) -- (\x{t}, \y{s})
        node [above, pos=0.4] {\(b\)}
        -- (\p{t})
        ;
      \end{tikzpicture}
      \caption{The parameterized arena D-NW-1.}
    \end{subfigure}
    \begin{subfigure}{.45\textwidth}
      \centering
      \begin{tikzpicture}[
          game,
          node distance = 25pt and 35pt,
        ]
        \node [state]                               (v)   {\(v\)};
        \node [state, above right=of v]             (x1)  {\(x_1\)};
        \node [state, below right=of v]             (x2)  {\(x_2\)};
        \node [state, right=of x1]                  (y1)  {\(y_1\)};
        \node [state, right=of x2]                  (y2)  {\(y_2\)};
        \node [state, accepting, below right=of y1] (t)   {\(t\)};
        \node [state, left=of v]                    (bad) {\(s\)};
        \node [below=of x2]                         (xi)  {\vdots};
        \node [below=of y2]                         (yi)  {\vdots};
        \node [state, below=of xi]                  (xn)  {\(x_n\)};
        \node [state, below=of yi]                  (yn)  {\(y_n\)};

        \path
        (v)   edge                  node [sloped]           {\(a_1\)}       (x1)
        edge                  node [sloped, pos=0.6]  {\(a_2\)}       (x2)
        edge                  node [sloped]           {\(c,
        \natsubs{\leq n}\)} (bad)
        edge                  node                    {\(c,
        \natsubs{> n}\)}    (t)
        (x1)  edge                  node                    {\(b,
        \{1\}\)}     (y1)
        (x2)  edge                  node [']                {\(b,
        \{2\}\)}     (y2)
        (y1)  edge                  node [sloped]           {\(b\)}         (t)
        (y2)  edge                  node [', sloped]        {\(b\)}         (t)
        (xn)  edge                  node                    {\(b,
        \{n\}\)}     (yn)
        ;

        \draw [rounded corners = 10pt]
        let
        \p{s} = (x1.west),
        \p{t} = (v.north),
        in
        (\p{s}) -- (\x{t}, \y{s})
        node [above, pos=0.4] {\(b, \natsubs{\neq 1}\)}
        -- (\p{t})
        ;
        \draw [rounded corners = 10pt]
        let
        \p{s} = (x2.west),
        \p{t} = (v.-60),
        in
        (\p{s}) -- (\x{t}, \y{s})
        node [below, pos=0.4] {\(b, \natsubs{\neq 2}\)}
        -- (\p{t})
        ;
        \draw [rounded corners = 10pt]
        let
        \p{s} = (v.-90),
        \p{t} = (xn.160),
        in
        (\p{s}) -- (\x{s}, \y{t})
        -- (\p{t})
        node [above, pos=0.4] {\(a_n\)}
        ;
        \draw [rounded corners = 10pt]
        let
        \p{s} = (xn.-160),
        \p{t} = (v.-130),
        in
        (\p{s}) -- (\x{t}, \y{s})
        node [below, pos=0.4] {\(b, \natsubs{\neq n}\)}
        -- (\p{t})
        ;
        \draw [rounded corners = 10pt]
        let
        \p{s} = (yn.east),
        \p{t} = (t.south),
        in
        (\p{s}) -- (\x{t}, \y{s})
        node [above, pos=0.4] {\(b\)}
        -- (\p{t})
        ;
      \end{tikzpicture}
      \caption{The parameterized arena D-W-1.}
    \end{subfigure}

    \begin{subfigure}{.45\textwidth}
      \centering
      \begin{tikzpicture}[
          game,
          node distance = 30pt and 55pt,
        ]
        \node [state]                               (v)   {\(v\)};
        \node [state, above=of v]              (s)   {\(s\)};
        \node [state, accepting, above right=of v]  (t v) {\(t\)};
        \node [state, right=of v]                   (xn)  {\(x_n\)};
        \node [state, right=of xn]                  (sn)  {\(s_n\)};
        \node [below=of xn]                         (dots){\(\vdots\)};
        \node [state, below=of dots]                (x1)  {\(x_1\)};
        \node [accepting, state, above right=of x1] (t1)  {\(t\)};
        \node [state, right=of x1]                  (s1)  {\(s\)};

        \path
        (v)   edge [bend left=15] node [sloped, near end] {\(a,
        \natsubs{>n}\)}           (t v)
        edge                node                    {\(a,
        \natsubs{\leq n}\)}       (s)
        edge                node                    {\(b_n\)}             (xn)
        (xn)  edge [bend left]    node                    {\(a,
        \natsubs{< n}\)}          (v)
        edge                node                    {\(a, \{n\}\)}
        (t v)
        edge                node                    {\(a, \natsubs{>
        n}\)}          (sn)
        (x1)  edge                node                    {\(a,
        \natsubs{> 1}\)}          (s1)
        ;

        \draw [rounded corners = 10pt]
        let
        \p{s} = (v.south),
        \p{t} = (x1.west),
        in
        (\p{s}) -- (\x{s}, \y{t})
        node [left, midway] {\(b_1, \natnonneg\)}
        -- (\p{t})
        ;
        \draw [rounded corners = 10pt]
        let
        \p{s} = (x1.30),
        \p{t} = (t1.west),
        in
        (\p{s}) -- ($(\x{s}, \y{t}) + (1, 0)$)
        -- (\p{t})
        node [above, near start] {\(a, \{1\}\)}
        ;
      \end{tikzpicture}     \caption{The parameterized arena D-NW-2.}
    \end{subfigure}
    \begin{subfigure}{.45\textwidth}
      \centering
      \begin{tikzpicture}[
          game,
          node distance = 30pt and 55pt,
        ]
        \node [state]                               (v)   {\(v\)};
        \node [state, accepting, above right=of v]  (t v) {\(t\)};
        \node [state, right=of v]                   (xn)  {\(x_n\)};
        \node [state, right=of xn]                  (sn)  {\(s_n\)};
        \node [below=of xn]                         (dots){\(\vdots\)};
        \node [state, below=of dots]                (x1)  {\(x_1\)};
        \node [accepting, state, above right=of x1] (t1)  {\(t\)};
        \node [state, right=of x1]                  (s1)  {\(s\)};

        \path
        (v)   edge [bend left]              node [sloped] {\(a,
        \natsubs{>n}\)}           (t v)
        edge [in=100, out=130, loop]  node          {\(a,
        \natsubs{\leq n}\)}       (v)
        edge                          node          {\(b_n\)}             (xn)
        (xn)  edge [bend left]              node          {\(a,
        \natsubs{< n}\)}          (v)
        edge                          node          {\(a, \{n\}\)}
        (t v)
        edge                          node          {\(a, \natsubs{>
        n}\)}          (sn)
        (x1)  edge                          node          {\(a,
        \natsubs{> 1}\)}          (s1)
        ;

        \draw [rounded corners = 10pt]
        let
        \p{s} = (v.south),
        \p{t} = (x1.west),
        in
        (\p{s}) -- (\x{s}, \y{t})
        node [left, midway] {\(b_1, \natnonneg\)}
        -- (\p{t})
        ;
        \draw [rounded corners = 10pt]
        let
        \p{s} = (x1.30),
        \p{t} = (t1.west),
        in
        (\p{s}) -- ($(\x{s}, \y{t}) + (1, 0)$)
        -- (\p{t})
        node [above, near start] {\(a, \{1\}\)}
        ;
      \end{tikzpicture}     \caption{The parameterized arena D-W-2.}
    \end{subfigure}

    \begin{subfigure}{\textwidth}
      \centering
      \begin{tikzpicture}[
          game,
          node distance = 25pt and 60pt,
        ]
        \node [state]                               (v)   {\(v\)};
        \node [state, above right=of v]             (x1)  {\(x_1\)};
        \node [state, below right=of v]             (x2)  {\(x_2\)};
        \node [state, right=of x1]                  (y1)  {\(y_1\)};
        \node [state, right=of x2]                  (y2)  {\(y_2\)};
        \node [state, accepting, below right=of y1] (t)   {\(t\)};
        \node [state, left=of v]                    (bad) {\(s\)};
        \node [below=of x2]                         (xi)  {\vdots};
        \node [below=of y2]                         (yi)  {\vdots};
        \node [state, below=of xi]                  (xn)  {\(x_n\)};
        \node [state, below=of yi]                  (yn)  {\(y_n\)};

        \path
        (v)   edge                  node [sloped]           {\(a\)}         (x1)
        edge                  node [sloped, pos=0.6]  {\(a\)}         (x2)
        edge                  node [sloped]           {\(c,
        \natsubs{\leq n}\)} (bad)
        edge                  node                    {\(c,
        \natsubs{> n}\)}    (t)
        (x1)  edge                  node                    {\(b,
        \{1\}\)}     (y1)
        (x2)  edge                  node [']                {\(b,
        \{2\}\)}     (y2)
        (y1)  edge                  node                    {\(b\)}         (t)
        (y2)  edge                  node [']                {\(b\)}         (t)
        (xn)  edge                  node                    {\(b,
        \{n\}\)}     (yn)
        ;

        \draw [rounded corners = 10pt]
        let
        \p{s} = (x1.west),
        \p{t} = (v.north),
        in
        (\p{s}) -- (\x{t}, \y{s})
        node [above, pos=0.4] {\(b, \natsubs{\neq 1}\)}
        -- (\p{t})
        ;
        \draw [rounded corners = 10pt]
        let
        \p{s} = (x2.west),
        \p{t} = (v.-60),
        in
        (\p{s}) -- (\x{t}, \y{s})
        node [below, pos=0.4] {\(b, \natsubs{\neq 2}\)}
        -- (\p{t})
        ;
        \draw [rounded corners = 10pt]
        let
        \p{s} = (v.-90),
        \p{t} = (xn.160),
        in
        (\p{s}) -- (\x{s}, \y{t})
        -- (\p{t})
        node [above, pos=0.4] {\(a\)}
        ;
        \draw [rounded corners = 10pt]
        let
        \p{s} = (xn.-160),
        \p{t} = (v.-130),
        in
        (\p{s}) -- (\x{t}, \y{s})
        node [below, pos=0.4] {\(b, \natsubs{\neq n}\)}
        -- (\p{t})
        ;
        \draw [rounded corners = 10pt]
        let
        \p{s} = (yn.east),
        \p{t} = (t.south),
        in
        (\p{s}) -- (\x{t}, \y{s})
        node [above, pos=0.4] {\(b\)}
        -- (\p{t})
        ;
      \end{tikzpicture}     \caption{The parameterized arena ND-NW.}
    \end{subfigure}
    \caption{The synthetic arenas used in the experiments of
      \Cref{sec:experiments}.
      We omit the constraint if it is
    \(\natnonneg\).}\label{app:fig:arenas:experiments}
  \end{figure}

  \begin{table}
    \caption{Time to construct the lattice for the games.
      Timeout is set to five minutes.
      Since D-W-1 and ND-NW (resp.\ D-NW-2 and D-W-2) have the same
      number of edges and the
      same constraints appearing on them, they are grouped together in
    this table.}\label{app:tab:experiments:lattice:games}
    \centering
    \footnotesize
    \setlength{\tabcolsep}{5pt}
    \begin{tabular}{@{}
        lrr
        @{\extracolsep{2\tabcolsep}}r
        @{\extracolsep{\tabcolsep}}r
      @{}}
      \toprule
      \multicolumn{3}{c}{Game} & \multicolumn{2}{c}{Lattice}
      \\
      \cmidrule{1-3}\cmidrule{4-5}
      Name & Value of \(n\) & Actual size & Size & Construction time
      \\
      \midrule
      \multirow[c]{3}{*}{D-NW-1} & 10 & 22 & 2048 & 2s 934ms \\
      & 11 & 24 & 4096 & 13s 732ms \\
      & 12 & 26 & 8192 & 66s 148ms \\
      \midrule
      \multirow[c]{3}{*}{D-W-1 and ND-NW} & 10 & 23 & 2048 & 2s 949ms \\
      & 11 & 25 & 4096 & 14s 172ms \\
      & 12 & 27 & 8192 & 68s 391ms \\
      \midrule
      \multirow[c]{3}{*}{D-NW-2 and D-W-2} & 10 & 13 & 2048 & 3s 115ms \\
      & 11 & 14 & 4096 & 14s 871ms \\
      & 12 & 15 & 8192 & 71s 543ms \\
      \bottomrule
    \end{tabular}
  \end{table}

  \begin{table}
    \caption{Number of iterations and generated sets for the \(\W\)
      and \(\WAlt\)
      methods, and the number of explored sub-games for the DFS algorithm
    from~\cite{BertrandBM19} on the synthetic arenas.}
    \label{app:tab:experiments:more}
    \centering
    \footnotesize
    \setlength{\tabcolsep}{5pt}
    \begin{tabular*}{\textwidth}{
        @{}
        lrr
        @{\extracolsep{\fill}}r@{\extracolsep{\tabcolsep}}r
        @{}r@{\extracolsep{\tabcolsep}}r
        @{\extracolsep{\fill}}r
        @{}
      }
      \toprule
      \multicolumn{3}{c}{Game} &
      \multicolumn{2}{c}{\(\W\)} &
      \multicolumn{2}{c}{\(\WAlt\)} &
      \multicolumn{1}{c}{Algorithm from~\cite{BertrandBM19}}
      \\
      \cmidrule{1-3}\cmidrule{4-5}\cmidrule{6-7}\cmidrule{8-8}
      Name & Value of \(n\) & Actual size &
      Iterations & Sets &
      Iterations & Sets &
      Sub-games
      \\
      \midrule
      \multirow[c]{6}{*}{D-NW-1} & 10 & 22 & 22 & 1024 & 22 & 1024 & 19728201 \\
      & 11 & 24 & 24 & 2048 & 24 & 2048 & timeout \\
      & 12 & 26 & 26 & 4096 & 26 & 4096 & timeout \\
      & 13 & 28 & timeout & timeout & 28 & 8192 & timeout \\
      & 14 & 30 & timeout & timeout & 30 & 16384 & timeout \\
      & 15 & 32 & timeout & timeout & 32 & 32768 & timeout \\
      \midrule
      \multirow[c]{6}{*}{D-W-1} & 10 & 23 & 22 & 1024 & 22 & 1024 & 31 \\
      & 11 & 25 & 24 & 2048 & 24 & 2048 & 34 \\
      & 12 & 27 & 26 & 4096 & 26 & 4096 & 37 \\
      & 13 & 29 & timeout & timeout & 28 & 8192 & 40 \\
      & 14 & 31 & timeout & timeout & 30 & 16384 & 43 \\
      & 15 & 33 & timeout & timeout & 32 & 32768 & 46 \\
      \midrule
      \multirow[c]{6}{*}{D-NW-2} & 10 & 13 & 21 & 1024 & 21 & 1024 & 135 \\
      & 11 & 14 & 23 & 2048 & 23 & 2048 & 169 \\
      & 12 & 15 & 25 & 4096 & 25 & 4096 & 204 \\
      & 13 & 16 & timeout & timeout & 27 & 8192 & 2861 \\
      & 14 & 17 & timeout & timeout & 29 & 16384 & 3962 \\
      & 15 & 18 & timeout & timeout & 31 & 32768 & 5064 \\
      \midrule
      \multirow[c]{6}{*}{D-W-2} & 10 & 13 & 22 & 2047 & 22 & 2047 & 34 \\
      & 11 & 14 & 24 & 4095 & 24 & 4095 & 35 \\
      & 12 & 15 & 26 & 8191 & 26 & 8191 & 103 \\
      & 13 & 16 & timeout & timeout & 28 & 16383 & 1101 \\
      & 14 & 17 & timeout & timeout & 30 & 32767 & 1102 \\
      & 15 & 18 & timeout & timeout & timeout & timeout & 1103 \\
      \midrule
      \multirow[c]{6}{*}{ND-NW} & 10 & 23 & 3 & 12 & 3 & 12 & 25963502 \\
      & 11 & 25 & 3 & 13 & 3 & 13 & timeout \\
      & 12 & 27 & 3 & 14 & 3 & 14 & timeout \\
      & 13 & 29 & timeout & timeout & 3 & 15 & timeout \\
      & 14 & 31 & timeout & timeout & 3 & 16 & timeout \\
      & 15 & 33 & timeout & timeout & 3 & 17 & timeout \\
      \bottomrule
    \end{tabular*}
  \end{table}

  \paragraph{Arenas obtained from quantified Boolean formulas.}

  \Crefrange{fig:app:experiments:solve:qbf:4}{fig:app:experiments:solve:qbf:6}
  show the time and the memory used by the three algorithms on the 50
  randomly generated QBF instances with
  four, five, and six variables, respectively.
  While \(\W\) and \(\WAlt\) manage to finish all formulas in time (\ie, never
  time out), the DFS
  algorithm of \cite{BertrandBM19} sometimes times out when the number of
  variables is greater or equal to five. The number of timeouts is also
  provided in the figures, when relevant.
  Finally, \Cref{app:tab:experiments:lattice:qbf} gives the time needed to
  construct \(\latticeK\) from the QBF instances.

  \begin{figure}[ht]
    \centering
    \begin{subfigure}{\textwidth}
      \centering
      \begin{tikzpicture}[
          /pgfplots/table/x expr={(\thisrow{Variables}==4 &&
          \thisrow{Clauses}>=60)?\thisrow{Clauses}:nan},
        ]
        \begin{axis}[
            footnotesize,
            width = 1\textwidth,
            height = 120pt,
            axis x line = bottom,
            axis y line = left,
            xlabel = {Number of clauses},
            xlabel near ticks,
            ylabel = {Time (ms)},
            ylabel near ticks,
            ylabel shift = {-5pt},
            ymajorgrids = true,
            xtick = data,
            enlarge y limits = 0.2,
            enlarge x limits = 0.15,
            nodes near coords = {\pgfmathprintnumber[1000
            sep={\,}]\pgfplotspointmeta},
          ]
          \addplot+ [
            Red,
            only marks,
            mark=*,
          ]
          table [
            y expr = \thisrow{W time total (ms)},
            col sep = comma,
          ]
          {qbf_solving.csv}
          ;
          \addplot+ [
            Blue,
            only marks,
            mark=triangle*,
            nodes near coords style = {below},
          ]
          table [
            y expr = \thisrow{WAlt time (ms)},
            col sep = comma,
          ]
          {qbf_solving.csv}
          ;
          \addplot+ [
            Indigo,
            only marks,
            mark=square*,
          ]
          table [
            y expr = \thisrow{Explicit time (ms)},
            col sep = comma,
          ]
          {qbf_solving.csv}
          ;
        \end{axis}
      \end{tikzpicture}
      \caption{Time. The values for \(\WAlt\) are given below the blue
      triangles, while values for \(\W\) are above the red circles.}
    \end{subfigure}

    \begin{subfigure}{\textwidth}
      \begin{tikzpicture}[
          /pgfplots/table/x expr={(\thisrow{Variables}==4 &&
          \thisrow{Clauses}>=60)?\thisrow{Clauses}:nan},
        ]
        \begin{axis}[
            footnotesize,
            width = 1\textwidth,
            height = 120pt,
            axis x line = bottom,
            axis y line = left,
            xlabel = {Number of clauses},
            xlabel near ticks,
            ylabel = {Memory (MB)},
            ylabel near ticks,
            ylabel shift = {-5pt},
            ymajorgrids = true,
            xtick = data,
            enlarge y limits = 0.3,
            enlarge x limits = 0.1,
            nodes near coords={\pgfmathprintnumber[1000
            sep={\,}]\pgfplotspointmeta},
          ]
          \addplot+ [
            Red,
            only marks,
            mark=*,
            nodes near coords style = {left},
          ]
          table [
            y expr = \thisrow{W memory (kB)} / 1000,
            col sep = comma,
          ]
          {qbf_memory.csv}
          ;
          \addplot+ [
            Blue,
            only marks,
            mark=triangle*,
          ]
          table [
            y expr = \thisrow{WAlt memory (kB)} / 1000,
            col sep = comma,
          ]
          {qbf_memory.csv}
          ;
          \addplot+ [
            Indigo,
            only marks,
            mark=square*,
            nodes near coords style = {below},
          ]
          table [
            y expr = \thisrow{Explicit memory (kB)} / 1000,
            col sep = comma,
          ]
          {qbf_memory.csv}
          ;
        \end{axis}
      \end{tikzpicture}
      \caption{Memory.}
    \end{subfigure}

    \caption{Time and memory taken by \(\W\), \(\WAlt\), and the DFS
      algorithm on the randomly generated QBF instances, with four
      variables. All algorithms could finish in time over all the
      instances. Red circles are for \(\W\), blue triangles for
      \(\WAlt\), and purple squares for the DFS algorithm. Values are
    displayed next to the symbol.}\label{fig:app:experiments:solve:qbf:4}
  \end{figure}

  \begin{figure}[ht]
    \begin{subfigure}{\textwidth}
      \centering
      \begin{tikzpicture}[
          /pgfplots/table/x expr={(\thisrow{Variables}==5 &&
          \thisrow{Clauses}>=60)?\thisrow{Clauses}:nan},
        ]
        \begin{axis}[
            footnotesize,
            width = 1\textwidth,
            height = 120pt,
            axis x line = bottom,
            axis y line = left,
            xlabel = {Number of clauses},
            xlabel near ticks,
            ylabel = {Time (ms)},
            ylabel near ticks,
            ylabel shift = {-5pt},
            ymajorgrids = true,
            xtick = data,
            enlarge y limits = 0.2,
            enlarge x limits = 0.15,
            nodes near coords={\pgfmathprintnumber[1000
            sep={\,}]\pgfplotspointmeta},
          ]
          \addplot+ [
            Red,
            only marks,
            mark=*,
          ]
          table [
            y expr = \thisrow{W time total (ms)},
            col sep = comma,
          ]
          {qbf_solving.csv}
          ;
          \addplot+ [
            Blue,
            only marks,
            mark=triangle*,
            nodes near coords style = {below},
          ]
          table [
            y expr = \thisrow{WAlt time (ms)},
            col sep = comma,
          ]
          {qbf_solving.csv}
          ;
          \addplot+ [
            Indigo,
            only marks,
            mark=square*,
          ]
          table [
            y expr = \thisrow{Explicit time (ms)},
            col sep = comma,
          ]
          {qbf_solving.csv}
          ;
        \end{axis}
      \end{tikzpicture}
      \caption{Time. The values for \(\WAlt\) are given below the blue
      triangles, while values for \(\W\) are above the red circles.}
    \end{subfigure}

    \begin{subfigure}{\textwidth}
      \centering
      \begin{tikzpicture}[
          /pgfplots/table/x expr={(\thisrow{Variables}==5 &&
          \thisrow{Clauses}>=60)?\thisrow{Clauses}:nan},
        ]
        \begin{axis}[
            footnotesize,
            width = 1\textwidth,
            height = 120pt,
            axis x line = bottom,
            axis y line = left,
            xlabel = {Number of clauses},
            xlabel near ticks,
            ylabel = {Memory (MB)},
            ylabel near ticks,
            ylabel shift = {-5pt},
            ymajorgrids = true,
            xtick = data,
            enlarge y limits = 0.3,
            enlarge x limits = 0.1,
            nodes near coords={\pgfmathprintnumber[1000
            sep={\,}]\pgfplotspointmeta},
          ]
          \addplot+ [
            Red,
            only marks,
            mark=*,
          ]
          table [
            y expr = \thisrow{W memory (kB)} / 1000,
            col sep = comma,
          ]
          {qbf_memory.csv}
          ;
          \addplot+ [
            Blue,
            only marks,
            mark=triangle*,
            nodes near coords style = {below},
          ]
          table [
            y expr = \thisrow{WAlt memory (kB)} / 1000,
            col sep = comma,
          ]
          {qbf_memory.csv}
          ;
          \addplot+ [
            Indigo,
            only marks,
            mark=square*,
            nodes near coords style = {below},
          ]
          table [
            y expr = \thisrow{Explicit memory (kB)} / 1000,
            col sep = comma,
          ]
          {qbf_memory.csv}
          ;
        \end{axis}
      \end{tikzpicture}
      \caption{Memory.}
    \end{subfigure}

    \begin{subfigure}{\textwidth}
      \centering
      \begin{tikzpicture}[
          /pgfplots/table/x expr={(\thisrow{Variables}==5 &&
          \thisrow{Clauses}>=60)?\thisrow{Clauses}:nan},
        ]
        \begin{axis}[
            footnotesize,
            width = 1\textwidth,
            height = 120pt,
            axis x line = bottom,
            axis y line = left,
            xlabel = {Number of clauses},
            xlabel near ticks,
            ylabel = {Number of timeouts},
            ylabel near ticks,
            ylabel shift = {-5pt},
            ymajorgrids = true,
            xtick = data,
            enlarge y limits = 0.2,
            enlarge x limits = 0.1,
            nodes near coords,
          ]
          \addplot+ [
            Indigo,
            only marks,
            mark=square*,
          ]
          table [
            y expr = \thisrow{Explicit timeouts},
            col sep = comma,
          ]
          {qbf_solving.csv}
          ;
        \end{axis}
      \end{tikzpicture}
      \caption{Timeouts.}
    \end{subfigure}
    \caption{Time and memory taken by \(\W\), \(\WAlt\), and the DFS
      algorithm on the randomly generated QBF instances, with five
      variables. Both \(\W\) and \(\WAlt\) could finish in time over all the
      instances. The last figure gives the number of timeouts for the
      DFS algorithm. Red circles are for \(\W\), blue triangles for
      \(\WAlt\), and purple squares for the DFS algorithm. Values are
    displayed  next to the symbol.}\label{fig:app:experiments:solve:qbf:5}
  \end{figure}

  \begin{figure}[ht]
    \centering
    \begin{subfigure}{\textwidth}
      \centering
      \begin{tikzpicture}[
          /pgfplots/table/x expr={(\thisrow{Variables}==6 &&
          \thisrow{Clauses}>=60)?\thisrow{Clauses}:nan},
        ]
        \begin{axis}[
            footnotesize,
            width = 1\textwidth,
            height = 120pt,
            axis x line = bottom,
            axis y line = left,
            xlabel = {Number of clauses},
            xlabel near ticks,
            ylabel = {Time (ms)},
            ylabel near ticks,
            ylabel shift = {-5pt},
            ymajorgrids = true,
            xtick = data,
            enlarge y limits = 0.2,
            enlarge x limits = 0.15,
            nodes near coords={\pgfmathprintnumber[1000
            sep={\,}]\pgfplotspointmeta},
          ]
          \addplot+ [
            Red,
            only marks,
            mark=*,
          ]
          table [
            y expr = \thisrow{W time total (ms)},
            col sep = comma,
          ]
          {qbf_solving.csv}
          ;
          \addplot+ [
            Blue,
            only marks,
            mark=triangle*,
            nodes near coords style = {below},
          ]
          table [
            y expr = \thisrow{WAlt time (ms)},
            col sep = comma,
          ]
          {qbf_solving.csv}
          ;
          \addplot+ [
            Indigo,
            only marks,
            mark=square*,
          ]
          table [
            y expr = \thisrow{Explicit time (ms)},
            col sep = comma,
          ]
          {qbf_solving.csv}
          ;
        \end{axis}
      \end{tikzpicture}
      \caption{Time. The values for \(\WAlt\) are given below the blue
      triangles, while values for \(\W\) are above the red circles.}
    \end{subfigure}

    \begin{subfigure}{\textwidth}
      \centering
      \begin{tikzpicture}[
          /pgfplots/table/x expr={(\thisrow{Variables}==6 &&
          \thisrow{Clauses}>=60)?\thisrow{Clauses}:nan},
        ]
        \begin{axis}[
            footnotesize,
            width = 1\textwidth,
            height = 120pt,
            axis x line = bottom,
            axis y line = left,
            xlabel = {Number of clauses},
            xlabel near ticks,
            ylabel = {Memory (MB)},
            ylabel near ticks,
            ylabel shift = {-5pt},
            ymajorgrids = true,
            xtick = data,
            enlarge y limits = 0.3,
            enlarge x limits = 0.1,
            nodes near coords={\pgfmathprintnumber[1000
            sep={\,}]\pgfplotspointmeta},
          ]
          \addplot+ [
            Red,
            only marks,
            mark=*,
          ]
          table [
            y expr = \thisrow{W memory (kB)} / 1000,
            col sep = comma,
          ]
          {qbf_memory.csv}
          ;
          \addplot+ [
            Blue,
            only marks,
            mark=triangle*,
            nodes near coords style = {above},
          ]
          table [
            y expr = \thisrow{WAlt memory (kB)} / 1000,
            col sep = comma,
          ]
          {qbf_memory.csv}
          ;
          \addplot+ [
            Indigo,
            only marks,
            mark=square*,
            nodes near coords style = {below},
          ]
          table [
            y expr = \thisrow{Explicit memory (kB)} / 1000,
            col sep = comma,
          ]
          {qbf_memory.csv}
          ;
        \end{axis}
      \end{tikzpicture}
      \caption{Memory.}
    \end{subfigure}

    \begin{subfigure}{\textwidth}
      \centering
      \begin{tikzpicture}[
          /pgfplots/table/x expr={(\thisrow{Variables}==6 &&
          \thisrow{Clauses}>=60)?\thisrow{Clauses}:nan},
        ]
        \begin{axis}[
            footnotesize,
            width = 1\textwidth,
            height = 120pt,
            axis x line = bottom,
            axis y line = left,
            xlabel = {Number of clauses},
            xlabel near ticks,
            ylabel = {Number of timeouts},
            ylabel near ticks,
            ylabel shift = {-5pt},
            ymajorgrids = true,
            xtick = data,
            ymin = 0,
            ymax = 50,
            ytick = {0, 25, 50},
            enlarge y limits = 0.2,
            enlarge x limits = 0.1,
            nodes near coords,
          ]
          \addplot+ [
            Indigo,
            only marks,
            mark=square*,
          ]
          table [
            y expr = \thisrow{Explicit timeouts},
            col sep = comma,
          ]
          {qbf_solving.csv}
          ;
        \end{axis}
      \end{tikzpicture}
      \caption{Timeouts.}
    \end{subfigure}
    \caption{Time and memory taken by \(\W\), \(\WAlt\), and the DFS
      algorithm on the randomly generated QBF instances, with six
      variables. Both \(\W\) and \(\WAlt\) could finish in time over all the
      instances. The last figure gives the number of timeouts for the
      DFS algorithm. Red circles are for \(\W\), blue triangles for
      \(\WAlt\), and purple squares for the DFS algorithm. Values are
    displayed  next to the symbol.}\label{fig:app:experiments:solve:qbf:6}
  \end{figure}

  \begin{table}[ht]
    \caption{Time to construct the lattice for the QBF instances, with the
    QBF-tailored algorithm.}\label{app:tab:experiments:lattice:qbf}
    \centering
    \footnotesize
    \begin{tabular}{@{} l rr @{}}
      \toprule
      \multirow{2}{50pt}{Number of variables} & \multicolumn{2}{c}{Lattice}
      \\
      \cmidrule(l){2-3}
      & Size & Construction time
      \\
      \midrule
      4 & 512 & 0s 6ms \\
      5 & 2048 & 0s 48ms \\
      6 & 8192 & 0s 529ms \\
      7 & 32768 & 6s 882ms \\
      \bottomrule
    \end{tabular}
  \end{table}
  \end{document}